\documentclass{article}

\usepackage{microtype}
\usepackage{graphicx}
\usepackage{subfigure}
\usepackage{verbatim}
\usepackage{booktabs} 
\usepackage{algpseudocode}
\usepackage{amsmath}
\usepackage{mathpazo}
\usepackage{caption}
\usepackage{wrapfig}
\usepackage[hyphens]{url}
\usepackage{hyperref}
\usepackage{siunitx}
\usepackage{enumitem}
\usepackage{xfrac}
\usepackage{changepage}

\usepackage{ifthen}
\usepackage[rm={lining},sf={lining},tt={tabular,lining,monowidth}]{cfr-lm}

\usepackage{amssymb}
\usepackage{pifont}
\newcommand{\xmark}{\ding{55}}%

\usepackage{xcolor}
\usepackage{threeparttable}

\newcommand{\custompar}[1]{\vspace{.15cm}\noindent{\bf #1}\:}

\newcommand{\softmax}{\mathsf{softmax}} 
\newcommand{\logsoftmax}{\mathsf{logsoftmax}} 
\newcommand{\logitbias}{b} 
\newcommand{\logit}{\mathsf{logit}} 
\newcommand{\N}{\mathbf{N}} 
\newcommand{\api}{\mathcal{O}}

\newcommand{\vocabsize}{l}
\newcommand{\topk}{\mathsf{TopK}}
\newcommand{\Argtopk}{\mathsf{ArgMax}}

\newcommand{\x}{p}
\newcommand{\z}{\mathbf{z}}
\newcommand{\Q}{\mathbf{Q}}
\newcommand{\U}{\mathbf{U}}
\newcommand{\V}{\mathbf{V}}
\newcommand{\W}{\mathbf{W}}
\newcommand{\E}{\Et^\top}
\newcommand{\Vt}{\mathbf{V}^\top}
\newcommand{\bS}{\mathbf{S}}
\newcommand{\R}{\mathbb{R}}

\renewcommand{\H}{\mathbf{H}}
\newcommand{\M}{\mathbf{M}}
\newcommand{\tEt}{\tilde{\Et}}
\newcommand{\Et}{\W}
\newcommand{\calX}{\mathcal{X}}
\newcommand{\calP}{\mathcal{P}}
\newcommand{\br}[1]{\left({#1}\right)}

\newcommand{\bias}{\logitbias}

\newcommand{\diag}{\mathrm{diag}}
\newcommand{\one}{\mathbf{1}}

\newcommand{\nicett}[1]{{\fontfamily{lmvtt}\fontsize{10.5pt}{12pt}\selectfont #1}}
\newcommand{\ada}{\nicett{ada}}
\newcommand{\babbage}{\nicett{babbage}}
\newcommand{\babbagetwo}{\nicett{babbage-002}}
\newcommand{\gptthree}{\nicett{gpt-3.5}}
\newcommand{\gptturbo}{\nicett{gpt-3.5-turbo}}
\newcommand{\gptturbonov}{\nicett{gpt-3.5-turbo-1106}}
\newcommand{\gptturboinstruct}{\nicett{gpt-3.5-turbo-instruct}}

\newcommand{\calC}{\mathcal{C}}

\DeclareMathOperator*{\argmax}{arg\,max}

\usepackage[accepted]{icml2024}

\usepackage{amsmath}
\usepackage{amssymb}
\usepackage{mathtools}
\usepackage{amsthm}

\usepackage[capitalize,noabbrev]{cleveref}

\def\draft{1}

\if\draft
\newcommand{\cc}[1]{\textcolor{cyan}{CC: #1}}
\else
\newcommand{\cc}[1]{}
\fi

\theoremstyle{plain}
\newtheorem{theorem}{Theorem}[section]

\newtheorem{lemma}[theorem]{Lemma}

\theoremstyle{definition}

\theoremstyle{remark}

\usepackage[textsize=tiny]{todonotes}

\newcommand{\customtitle}{Stealing Part of a Production Language Model}

\icmltitlerunning{Stealing Part of a Production Language Model}

\begin{document}

\twocolumn[
\icmltitle{\customtitle}




\begin{icmlauthorlist}
\icmlauthor{Nicholas Carlini}{google}
\icmlauthor{Daniel Paleka}{ethz}
\icmlauthor{Krishnamurthy (Dj) Dvijotham}{google}
\icmlauthor{Thomas Steinke}{google}
\icmlauthor{Jonathan Hayase}{uw}
\icmlauthor{A. Feder Cooper}{google}
\icmlauthor{Katherine Lee}{google}
\icmlauthor{Matthew Jagielski}{google}
\icmlauthor{Milad Nasr}{google}
\icmlauthor{Arthur Conmy}{google}
\icmlauthor{Itay Yona}{google}
\icmlauthor{Eric Wallace}{oai}
\icmlauthor{David Rolnick}{mcgill}
\icmlauthor{Florian Tram\`{e}r}{ethz}
\end{icmlauthorlist}

\icmlaffiliation{google}{Google DeepMind}
\icmlaffiliation{ethz}{ETH Zurich}
\icmlaffiliation{mcgill}{McGill University}
\icmlaffiliation{uw}{University of Washington}
\icmlaffiliation{oai}{OpenAI}

\icmlcorrespondingauthor{Nicholas Carlini}{nicholas@carlini.com}

\icmlkeywords{Machine Learning, ICML}
 
\vskip 0.3in
]



\printAffiliationsAndNotice{}

\begin{abstract}
  We introduce the first model-stealing attack that extracts precise,
  nontrivial information from black-box production language models like Open\-AI's ChatGPT or Google's PaLM-2.
  Specifically, our attack recovers the \emph{embedding projection layer} (up to symmetries) 
  of a transformer model, given typical API access.
  For under \$20 USD, our attack extracts the entire projection matrix of Open\-AI's \ada{} and \babbage{} language models. We
  thereby confirm, for the first time, that these black-box models
  have a hidden dimension of 1024 and 2048, respectively.
  We also recover the exact hidden dimension size of the \gptturbo{} model,
  and estimate it would cost under \$2{,}000 in queries to recover the entire projection matrix.
  We conclude with potential defenses and mitigations,
  and discuss the implications of possible future work that
  could extend our attack.
\end{abstract}
 
\vspace{-1em}
\section{Introduction}

Little is publicly known about the inner workings of
today's most popular large language models, 
such as GPT-4, Claude 2, or Gemini.
The GPT-4 technical 
report states it ``contains no [...] details about the architecture (including model size), hardware, training compute,
dataset construction, training method, or similar''~\cite{openai2023gpt4}. 
Similarly, the PaLM-2 paper states that ``details of [the] model size and architecture are withheld from external publication''~\cite{anil2023palm}.
This secrecy is often ascribed to 
``the competitive landscape'' 
(because these models are expensive to train) and 
the ``safety implications of large-scale models''~\cite{openai2023gpt4}
(because it is easier to attack models when more information is available).
Nevertheless, while these models' weights and internal details are not publicly accessible,
the models themselves are exposed via APIs. 

In this paper we ask:
\emph{how much information can an adversary learn about a 
production language model by making queries to its API?}
This is the question studied by the field of \emph{model stealing}~\cite{tramer2016stealing}:  
the ability of an adversary to extract model weights by making queries its API.\looseness=-1

\textbf{Contributions.~~}
We introduce an attack that can be applied to
black-box language models,
and allows us to recover the complete \emph{embedding projection layer} 
of a transformer language model.
Our attack departs from prior approaches that reconstruct a model in a \emph{bottom-up} fashion, 
starting from the input layer. 
Instead, our attack operates \emph{top-down} and directly extracts the model's last layer. 
Specifically, we exploit the fact that the final layer of a language model projects from the hidden dimension to a (higher dimensional) logit vector. This final layer is thus low-rank, and by making targeted queries to a model's API, we can extract its embedding dimension or its final weight matrix.

Stealing this layer is useful for several reasons.
First, it reveals the \emph{width} of the transformer model, which is often correlated with its total parameter count.
Second, it slightly reduces the degree to which the model is a complete ``black-box'',
which so might be useful for future attacks.
Third, while our attack recovers only a (relatively small) part of the entire model,
the fact that it is at all possible to steal \emph{any} parameters of a
production model is surprising, and raises concerns that
extensions of this attack might be able to recover more information.
Finally, recovering the model's last layer (and thus hidden dimension) may reveal more global information about the model, such as relative size differences between different models.

Our attack is effective and efficient, and is
applicable to production models whose APIs expose full logprobs, or a ``logit bias''.
This included
Google's PaLM-2 and OpenAI's {GPT-4}~\citep{anil2023palm, openai2023gpt4};
after responsible disclosure, both APIs have implemented defenses
to prevent our attack or make it more expensive.
We extract the embedding layer of several OpenAI models with a  mean squared error
of $10^{-4}$ (up to unavoidable symmetries).
We apply a limited form of our attack to \gptthree{} at a cost of under \$200 USD and,
instead of recovering the full embedding layer, recover just the size of the embedding
dimension.
%

\textbf{Responsible disclosure.~~}
We shared our attack with all services we are aware of that are vulnerable to this attack.
We also shared our attack with several other popular services, even if they were
not vulnerable to our specific attack, because variants of our attack may
be possible in other settings.
We received approval from OpenAI prior to extracting the parameters of the last layers of their models,
worked with OpenAI to confirm our approach's efficacy, and then deleted all data associated with
the attack.
In response to our attack,
OpenAI and Google have both modified their APIs to introduce
mitigiations and defenses (like those that we suggest in Section~\ref{sec:defense}) to make it more difficult for adversaries to perform this attack.

\section{Related Work}

Model stealing attacks \cite{tramer2016stealing} aim to recover the functionality of a black-box model, and optimize for one of two objectives \cite{jagielski2020high}:
\begin{enumerate}[topsep=0pt, itemsep=-1pt]
    \item \emph{Accuracy}: the stolen model $\hat{f}$ should match the performance of the target model $f$ on some particular data domain.
    For example, if the target is an image classifier, we might want the stolen model to match the target's overall accuracy on ImageNet.
    \item \emph{Fidelity}: the stolen model $\hat{f}$ should be functionally equivalent to the target model $f$ on all inputs. That is, for any valid input $\x$, we want $\hat{f}(\x) \approx f(\x)$.
\end{enumerate}

In this paper, we focus on high-fidelity attacks.
Most prior high-fidelity attacks exploit specific properties
of deep neural networks with ReLU activations.
\citet{milli2019model} first showed that if an attacker can compute \emph{gradients} of a target two-layer ReLU model, then they can steal
a nearly bit-for-bit equivalent model.
\citet{jagielski2020high} observed that if the attacker only has query access to model outputs, they can approximate gradients with finite differences.
Subsequent work extended these attacks to efficiently extract deeper ReLU models~\citep{carlini2020cryptanalytic, rolnick2020reverse, shamir2023polynomial}.
Unfortunately, none of these approaches scale to production
language models, because they (1) accept tokens as inputs (and so performing finite
differences is intractable); (2) use activations other than ReLUs; (3) contain architectural components such as attention, 
layer normalization, residual connections, etc.~that current attacks cannot handle; (4) are orders-of-magnitude larger than prior extracted models; and (5) expose only limited-precision outputs.

Other attacks aim to recover more limited information, or assume a stronger adversary.
\citet{wei2020leaky} show that an adversary co-located on the same server as the LLM can recover the sizes of all hidden layers.
\citet{zanella2021grey} assume a model with a public pretrained encoder and a private final layer, and extract the final layer; our Section~\ref{sec:attack_detail} is quite similar to their method.
Others 
have attempted to recover model sizes by correlating performance
on published benchmarks with model sizes in academic papers~\cite{gao2021sizes}.

\vspace{-.1cm}
\section{Problem Formulation}

We study models that take a sequence of tokens drawn from a vocabulary $\mathcal{X}$ as input. 
Let $\calP\br{\calX}$ denote the space of probability distributions over $\calX$. We study parameterized models $f_\theta : \mathcal{X}^N \to \calP\br{\calX}$ that produce a probability distribution over the next output token,
given an input sequence of $N$ tokens. The model has the following structure:
\vspace{-.1cm}
\begin{align}
 f_\theta(\x) = \softmax(\Et\cdot  g_\theta(\x)),    \label{eq:transformer}
\end{align}
where $g_\theta : \mathcal{X}^N \to \R^h$
is another parameterized model that computes hidden states, 
$\Et$ is an $l \times h$ dimensional matrix (the \emph{embedding projection matrix}),
and $\softmax : \R^l \to [0,1]^l$ is the softmax function applied to the resulting \emph{logits}:
\vspace{-.2cm}
\[
\softmax(\z) = \left[\frac{e^{\z_{1}}}{\sum_{i=1}^l e^{\z_{i}}}, \dots, \frac{e^{\z_{l}}}{\sum_{i=1}^l e^{\z_{i}}} \right] \;.
\]
Note that the hidden dimension size is much smaller than the size of the token dictionary, i.e., $h \ll l$.
For example, LLaMA~\citep{touvron2023llama} chooses $h \in \{4096, 5120, 6656, 8192\}$ and $l = 32{,}000$,
and there is a recent trend towards increasingly large token sizes;
GPT-4, for example, has a $\approx$100,000 token vocabulary.
%
%

\custompar{Threat model.}
Throughout the paper, we assume that the adversary does not have any additional knowledge about the model parameters. We assume access to a model $f_\theta$,
hosted by a service provider and made available to users through
a query interface (API) $\api$.
We assume that $\api$ is a perfect oracle:
given an input sequence $\x$, it produces
$y=\api\br{\x}$ without leaking any other information about $f_\theta$ than what can be inferred from $(\x, y)$.
For example, the adversary cannot
infer anything about $f_\theta$ via timing side-channels or other details of the implementation of the query interface.

Different open-source and proprietary LLMs offer APIs with varying capabilities, which impact the ability to perform model extraction attacks and the choice of attack algorithm. 
A summary of the different APIs we study, and our motivation for doing so, is presented in Table \ref{tab:threat models_summary}. The logits API is a strawman threat model where the API provides logits for all tokens in the response to a given prompt. We begin with this toy setting, as the attack techniques we develop here can be reused in subsequent sections, where we will first reconstruct the logits from more limited information (e.g., log-probabilities for only the top few tokens) and then run the attack.

\begin{table}
\small
\centering
\caption{Summary of APIs}
\label{tab:threat models_summary}
\vspace{5pt}
\setlength{\tabcolsep}{4pt}
\begin{tabular}{@{} l l @{}}
\toprule
API & Motivation \\
\midrule
All Logits \S\ref{sec:mainwarmup} & Pedagogy \& basis for next attacks  \\
Top Logprobs, Logit-bias \S\ref{sec:logit_bias_attacks} & Current LLM APIs (e.g., OpenAI)\\
No logprobs, Logit-bias \S\ref{sec:logprob-free-appendix} & Potential future constrained APIs \\
\bottomrule
\end{tabular}
\end{table}
\setlength{\textfloatsep}{8pt}




\section{Extraction Attack for Logit-Vector APIs}\looseness=-1
\label{sec:mainwarmup}

In this section, we assume the
adversary can directly view the logits that feed into the softmax function for every token in the vocabulary (we will later relax this assumption), i.e., 
\[\api\br{\x} \gets \Et \cdot g_\theta\br{\x} \;.\]
We develop new attack techniques that allow us to perform high-fidelity extraction of (a small part of) a transformer. Section \ref{sec:attack_warmup} demonstrates how we can identify the hidden dimension $h$ using the logits API and Section \ref{sec:attack_detail} presents an algorithm that can recover the matrix $\Et$.

\subsection{Warm-up: Recovering Hidden Dimensionality}\label{sec:attack_warmup}

We begin with a simple attack that allows an adversary to recover the 
size of the hidden dimension of a language model by making queries to the oracle $\api$  (Algorithm~\ref{algo:hidden}).
The techniques we use to perform this attack will be the foundation for 
attacks that we further develop to perform complete extraction of the final
embedding projection matrix.

\setlength{\textfloatsep}{8pt}
\begin{algorithm}[h]
\caption{Hidden-Dimension Extraction Attack}
\label{algo:hidden}
\begin{algorithmic}[1]
\Require Oracle LLM $\api$ returning \textbf{logits} 

\State Initialize $n$ to an appropriate value greater than $h$
\State Initialize an empty matrix $\Q = \mathbf{0}^{n \times l}$
\For{$i = 1$ to $n$}
    \State $\x_i \gets \texttt{RandPrefix}()$ \Comment{\small{Choose a random prompt}}
    \State $\Q_i \gets \api(\x_i)$
\EndFor
\State $\lambda_1 \geq \lambda_2 \geq \dots \geq \lambda_n \gets \text{SingularValues}(\Q)$

\State $\text{count} \gets \argmax_i {\log \lVert\lambda_{i}\rVert - \log \lVert\lambda_{i+1}\rVert}$

\State \Return $\text{count}$
\end{algorithmic}
\end{algorithm}

\custompar{Intuition.}
Suppose we query a language model on a large number of different random prefixes.
Even though each output logit vector is an $l$-dimensional vector,
they all actually lie in a $h$-dimensional subspace because the embedding projection layer up-projects from $h$-dimensions.
Therefore, by querying the model ``enough'' (more than $h$ times) we will eventually observe
new queries are linearly dependent of past queries.
We can then compute the dimensionality of this subspace (e.g., with SVD)
and report this as the hidden dimensionality of the model.

\custompar{Formalization.}
The attack is based on the following straightforward mathematical result:
\begin{lemma}\label{lem:h_recovery}
Let $\Q\br{\x_1, \ldots \x_{n}} \in \R^{l \times n}$ denote the matrix with columns  $\api\br{\x_1}, \ldots, \api\br{\x_{n}}$ of query responses from the logit-vector API. Then 
\[h \geq \text{rank}\br{\Q\br{\x_1, \ldots \x_{n}}}.\]
Further, if the matrix with columns $g_\theta\br{\x_i}$ ($i=1, ..., n$) has rank $h$ and $\Et$ has rank $h$, then 
\[h = \text{rank}\br{\Q\br{\x_1, \ldots \x_{n}}}.\]
\end{lemma}
\begin{proof}

We have
$\Q = \Et \cdot \H$, where $\H$ is a $h \times n$ matrix whose columns are $g_\theta(\x_i)$ ($i = 1, \ldots, n$). Thus,
$h \geq \textrm{rank}\br{\Q}$. Further, if $\H$ has rank $h$ (with the second assumption), then $h = \textrm{rank}\br{\Q}$.
\end{proof}

\custompar{Assumptions.} In \Cref{lem:h_recovery}, we assume that both the matrix with columns $g_\theta\br{\x_i}$ and the matrix $\Et$ have rank $h$. These matrices have either $h$ rows or $h$ columns, so both have rank at most $h$.
Moreover, it is very unlikely that they have rank $<h$: this would require the distribution of $g_\theta\br{\x}$ to be fully supported on a subspace of dimension $<h$ across all $\x_i$ we query, or all $h \ll l$ columns of $\Et$ to lie in the same $(h-1)$ dimensional subspace of $\mathbb{R}^l$ (the output space of logits).
In practice we find this assumption holds for all larger models (\Cref{tab:model_comparison_open_source}) and when different normalization layers are used (\Cref{subapp:normalization_layers}).


\custompar{Practical considerations.} Since the matrix $\Q$ is not computed over the reals, but over 
floating-point numbers (possibly with precision as low as 16-bits or 8-bits
for production neural networks),
we cannot naively
take the rank to be the 
number of linearly independent rows. Instead, we use a practical 
\emph{numerical rank} of $\Q$,
where we order the singular values $\lambda_1 \ge \lambda_2 \ge \cdots \ge \lambda_n$, and identify the largest \emph{multiplicative} gap $\frac{\lambda_i}{\lambda_{i+1}}$ between consecutive singular values. 
A large multiplicative gap arises when we switch from large ``actual'' singular values to small singular values that arise from numerical imprecision.
Figure \ref{fig:why_paramfree_works} shows these gaps.
Algorithm~\ref{algo:hidden} describes this attack.
%


\begin{figure}[t]
    \centering
    \includegraphics[scale=.75]{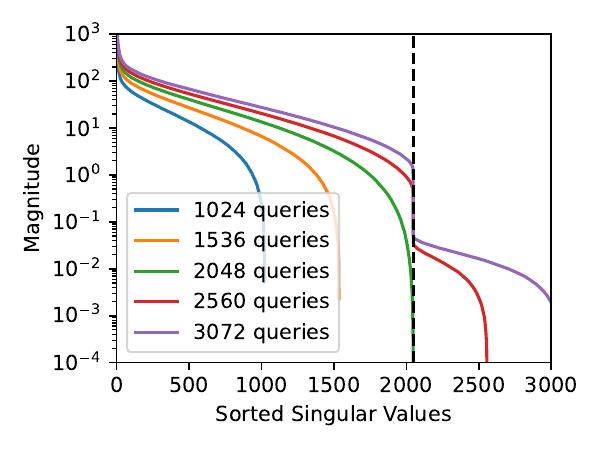}
    \vspace{-0.45cm}
    \caption{SVD can recover the hidden dimensionality of a 
    model when the final output layer dimension is greater than the hidden dimension. 
    Here we extract the hidden dimension (2048) of the Pythia 1.4B model.
    We can precisely identify the size by obtaining slightly over 2048 full logit vectors.}
    \label{fig:why_svd_works}
\end{figure}

\textbf{Experiments.~} In order to visualize the intuition behind this attack,
Figure~\ref{fig:why_svd_works} illustrates an attack against the Pythia-1.4b LLM.
Here, we plot the magnitude of the singular values of $\Q$ as we send an increasing number $n$ of queries to the model.
When we send fewer than 2048 queries it is impossible to identify
the dimensionality of the hidden space.
This is because $n < h$, and so the $n \times l$ dimensional matrix $\Q$ has
full rank and $n$ nontrivial singular values.
But once we make more than $2048$ queries to the model, and thus $n > h$,
the number of numerically significant singular values does not increase further; it is capped at exactly $2048$.

In Figure~\ref{fig:why_paramfree_works} we plot the difference (in log-space)
between subsequent singular values.
As we can see, the largest difference occurs at (almost exactly) the 2048th singular value---the true
hidden dimensionality of this model.

We now analyze the efficacy of this attack across a wider range of models: GPT-2 \cite{radford2019gpt2} Small and XL,
Pythia \cite{biderman2023pythia} 1.4B and 6.9B, and LLaMA \cite{touvron2023llama} 7B and 65B.
The results are in Table~\ref{tab:model_comparison_open_source}:
our attack recovers the embedding size nearly perfectly, with an error of 0 or 1 in five out of six cases.

\begin{figure}[t]
    \centering
    \includegraphics[scale=.75]{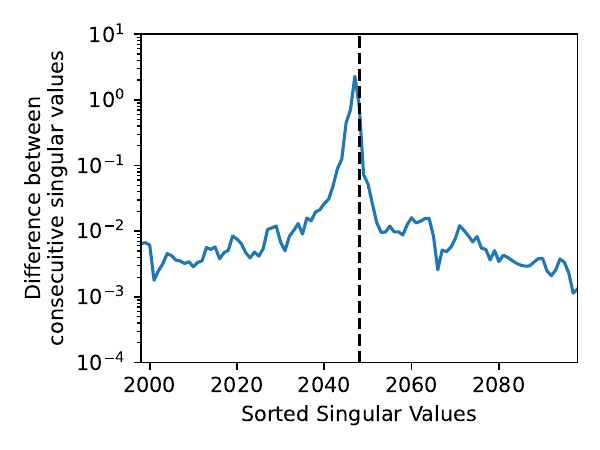}
    \vspace{-0.45cm}
    \caption{Our extraction attack recovers the hidden dimension 
    by identifying
    a sharp drop in singular values, visualized as a spike in the difference between consecutive singular values. On Pythia-1.4B, a 2048 dimensional model, the spike occurs at 2047 
    values.}
    \label{fig:why_paramfree_works}
    \vspace{1em}
\end{figure}

Our near perfect extraction has one exception: GPT-2 Small.
On this 768 dimensional model, our attack reports a hidden dimension of 757.
In Appendix~\ref{sec:whygpt2} we show that this ``failure'' is caused by GPT-2 actually having an effective
hidden dimensionality of $757$ despite having $768$ dimensions.

\paragraph{Cheaper Dimension Extraction}
Note that $l$ being exactly equal to the vocabulary size is not crucial.
Formally, taking only $l' < l$ rows of $\Q$ does not change the number of nonzero singular values, 
except in the unlikely case that the resulting submatrix is of smaller rank.
Hence, we can choose a subset of $l'$ tokens and extract the dimension from logits on these tokens alone, as long as $l' > h$.

\subsection{Full Layer Extraction (Up to Symmetries)}\label{sec:attack_detail}

We extend the attack from the prior section to recover the 
final output projection matrix $\Et$ that maps from the final hidden layer to the output logits.

\textbf{Method:~} Let $\mathbf{Q}$ be as defined in Algorithm \ref{algo:hidden}. Now rewrite $\Q = \U \cdot \mathbf{\Sigma} \cdot \Vt$ with SVD.
Previously we saw that the number of large enough singular values corresponded to the dimension of the model.
But it turns out that the matrix $\U$ actually directly represents (a linear transformation of) the final layer!
Specifically, we can show that $\U \cdot \mathbf{\Sigma}=\Et \cdot \mathbf{G}$ for some $h \times h$ matrix $\mathbf{G}$ in the following lemma.

\begin{lemma}
\label{lemma:recover-E-up-to-rotation}
In the logit-API threat model, under the assumptions of Lemma ~\ref{lem:h_recovery}: \textnormal{(i)}
The method above recovers $\tEt = \Et \cdot \mathbf{G}$ for some $\mathbf{G} \in \R^{h \times h}$; \textnormal{(ii)}
With the additional assumption that $g_{\theta}(p)$ is a transformer with residual connections, it is impossible to extract $\Et$ exactly.
\end{lemma}
\emph{Proof.~}
See Appendix~\ref{sec:proof_of_42}. \hfill$\square$
%



Note that we could also use $\Q = \Et \cdot \mathbf{G}$  for $n=l$. The SVD construction above gains numerical precision if $n > l$. 

\paragraph{Experiments.} For the six models considered previously, 
we evaluate the attack success rate by comparing the
root mean square (RMS) between our extracted matrix $\tEt = \U \cdot \mathbf{\Sigma}$ and the
actual weight matrix, after allowing for a $h \times h$ affine
transformation.
Concretely, we solve the least squares system $\tEt \cdot \mathbf{G} \approx \Et$ for $\mathbf{G}$, which reduces to $h$ linear least squares problems, each with $l$ equations and $h$ unknowns. Then, we report the RMS of $\Et$ and $\tEt \cdot \mathbf{G}$.
The results are in Table \ref{tab:model_comparison_open_source}.
%
%
As a point of reference, the RMS between a randomly initialized
model and the actual weights is $2\cdot10^{-2}$, over $100$--$500\times$
higher than the error of our reconstruction.

In Appendices \ref{sec:proof_of_42} and \ref{app:recovery-orthogonal}, we show that reconstruction is possible up an \emph{orthogonal} transformation (approximately $h^2/2$ missing parameters, as opposed to $h^2$ for reconstruction up to an affine transformation), and that this is tight under some formal assumptions.
However, we only have an efficient algorithm for reconstruction up to affine transformations.

\begin{table}
\small
\centering
\caption{Our attack succeeds across a range of open-source models,
at both stealing the model size, and also at reconstructing the output projection matrix (up to invariances; we show the root MSE).}
\label{tab:model_comparison_open_source}
\vspace{5pt}
\begin{tabular}{@{}l r r r@{}}
\toprule
Model & \hspace{-1.5em} Hidden Dim & Stolen Size & $\Et$ RMS \\
\midrule
GPT-2 Small  (fp32)  &  $768$ & 757 $\pm$ 1  & $4 \cdot 10^{-4}$ \\
GPT-2 XL     (fp32)  & $1600$ & 1599 $\pm$ 1  & $6 \cdot 10^{-4}$ \\
Pythia-1.4   (fp16)  & $2048$ & 2047 $\pm$ 1  & $3 \cdot 10^{-5}$  \\
Pythia-6.9   (fp16)  & $4096$ & 4096 $\pm$ 1  & $4 \cdot 10^{-5}$ \\
LLaMA 7B     (fp16)  & $4096$ & 4096 $\pm$ 2  & $8 \cdot 10^{-5}$ \\
LLaMA 65B    (fp16)  & $8192$ & 8192 $\pm$ 2  & $5 \cdot 10^{-5}$ \\
\bottomrule
\end{tabular}
\vspace{0.5em}
\end{table}

\section{Extraction Attack for Logit-Bias APIs}\label{sec:logit_bias_attacks}

The above attack makes a significant assumption: 
that the adversary can directly observe
the complete logit vector for each input.
In practice, this is not true: no production model we are aware of
provides such an API.
Instead, for example, they provide a way for users to get the top-$K$ (by logit) token log probabilities. 
In this section we address this challenge.

\subsection{Description of the API} \label{Sec:logitbiasAPI}
In this section we develop attacks for APIs that
return log probabilities for the top $K$ tokens (sorted by logits), 
and where the user can specify a real-valued bias $\bias \in \R^{|\calX|}$ (the ``logit bias'')
to be added to the logits for specified tokens before the softmax, i.e., 
\begin{align*}
    & \api(\x, \bias) \gets  \topk\br{\logsoftmax\br{\Et g_\theta(\x) + \bias}} \\
    &=\! \topk\!\left(\!\Et g_\theta(\x) \!+\! \bias \!-\! \log\!\left(\!\sum_i \exp\!\left(\!\Et g_\theta(\x) \!+\! \bias\right)_i\!\right)\!\cdot\!\mathbf{1}\!\right)\!.
\end{align*}
where $\topk\br{\z}$ returns the $K$ highest entries of $\z \in \R^l$ and their indices.
Many APIs (prior to this paper) provided such an option for their 
state-of-the-art models \cite{openai:create-chat-completion,google:changelog-1.38.0}.
In particular, the OpenAI API supports modifying logits for at most $300$ tokens,
  and the logit bias for each token is restricted to the range $[-100, 100]$~\citep{openai:logit-bias}.

All that remains is to show that we can uncover the full logit vector for distinct prompt queries through this API. In this section, we develop techniques for this purpose. 
Once we have recovered multiple complete logit vector, we can run the attack from Section \ref{sec:attack_detail} without modification.


\subsection{Evaluation Methodology}

Practical attacks must be \emph{efficient}, both to keep the cost of extraction manageable and to bypass any rate limiters or other filters in the APIs. We thus begin with two cost definitions that we use to measure the efficacy of our attack.
%

\textbf{Token cost:} the number of tokens the adversary sends to (or receives
from) the model during the attack.
Most APIs charge users per-token, so
this metric represents the monetary cost of an attack
(after scaling by the token cost).%
%

\textbf{Query cost:} the total duration of the attack.
Most APIs place a limit on the number of queries an adversary can
make in any given interval, and so some attacks may
be faster but cost more (by sending more tokens per query).

In the remainder of this section we develop several attacks under varying
attack assumptions and optimizing for either \emph{token cost}, \emph{query cost},
or both.

\subsection{Extraction Attack for Top-5 Logit Bias APIs	}\label{sec:topk-extraction-main}
We develop a technique to compute the logit vector for any prefix $p$ via a sequence of queries with varying logit biases.

To begin, \textbf{suppose that the API returned the top $K$ logits}. 
Then we could recover the complete logit vector for an arbitrary prompt $p$ by cycling through different choices for the logit bias and measuring the top-$k$ logits each time. In particular, for an API with top-5 logits we can send a sequence of queries 
\[
\api(p, \logitbias_{k}\!=\!\logitbias_{k+1}\!=\!\dots\!=\!\logitbias_{k+4}\!=\!B) \text{, for } k\! \in\! \{\!0,\! 5,\! 10, \!\dots\!,\! |\mathcal{X}|\!\}
\]
with a large enough $B$.
Each query thus promotes five different tokens $\{k, k+1, \ldots, k+4\}$ into 
the top-5, which allows us to observe their logits. 
By subtracting the bias $B$ and merging answers from all of these queries, we recover the entire logit vector.

Unfortunately, we cannot use this attack directly
because all production APIs we are aware of return \emph{logprobs} (the log of the softmax output of the model) instead of the logits $z_i$.
%
%
The problem now is that when we apply a logit bias $B$ to the $i$-th token and observe that token's logprob, 
we get the value
\[
y_i^B = z_i + B - \log \big(\sum_{j\neq i} \exp(z_j) + \exp(z_i + B) \big)
\]
where $z_i$ are the original logits. 
We thus get an additional bias-dependent term which we need to deal with. We propose two approaches.

Our first approach relies on a common ``reference'' token that lets us learn the relative difference between all logits (this is the best we can hope for, since the softmax is invariant under additive shifts to the logits).
Suppose the top token for a prompt is $R$, and we want to learn the relative difference between the logits of tokens $i$ and $R$. We add a large bias $B$ to token $i$ to push it to the top-5, and then observe the logprobs of both token $i$ and $R$. We have:
\[
y_R^B - y_i^B - B = z_R - z_i \,.
\]
Since we can observe 5 logprobs, we can compare the reference token $R$ to four tokens per query, by adding a large bias that pushes all four tokens into the top 5 (along with the reference token).
We thus issue a sequence of queries 
\[
\api(p, \logitbias_{i}=\logitbias_{i+1}=\logitbias_{i+2}=\logitbias_{i+3}=B)
\]
for $i \in \{0, 4, 8, \cdots, |\mathcal{X}|\}$.
This recovers the logits up to the free parameter $z_R$ that we set to $0$.

%

\textbf{Query cost.}
This attack reveals the value of K-1 logits with each query to the model
(the $K$-th being used as a reference point),
for a cost of $1/(K-1)$ queries per logit.

In Appendix \ref{sec:linear_reconstruction} we present a second, more sophisticated method that allows us to recover $K$ logits per query,  i.e., a cost of $1/K$, by viewing each logprob we receive as a linear constraint on the original logits.

\textbf{Token cost.}
%
%
Recall that our attack requires that we learn the logits for several distinct prompts;
and so each prompt must be at least one token long.
Therefore, this attack costs at least two tokens per query (one input and one output), or a cost of
$1/2$ for each token of output.
But, in practice, many models (like \gptturbo{}) include a
few tokens of overhead along with every single query. 
This increases the token cost per logit to ${2 + \Delta} \over 4$ where $\Delta$ is the
number of overhead tokens; for \gptturbo{} we report $\Delta=7$.

\paragraph{An improved cost-optimal attack.}
It is possible to generalize the above attack to improve \emph{both}
the query cost and token cost.
Instead of issuing queries to the model that reveal 4 or 5 logit values for a single generated token, we might instead hope to be
able to send a multi-token query
$[p_0 \quad p_1 \quad p_2 \dots p_n]$
and then ask for the logprob vector for
each prefix of the prompt
$[p_0]$, $[p_0 \quad p_1]$, $[p_0 \quad p_1 \quad p_2]$ etc.
OpenAI's API did allow for queries of this
form in the past, by providing logprobs for
\emph{prompt} tokens as well as generated tokens by combining the \nicett{logprob} and \nicett{echo} parameters;
this option has since been removed.

Now, it is only possible to view logprobs of \emph{generated} tokens.
And since only the very last token is generated, we can only view four
logprobs for this single longer query.
This, however, presents a potential approach to reduce the query and token cost:
if there were some way to cause the model to emit
a specific sequence of tokens $[p_{n+1} \quad p_{n+2} \quad \dots \quad p_{n+m}]$, then we could inspect the logprob vector of each generated token.

We achieve this as follows: we fix a token $x$ and four other tokens,
and force the model to emit $[x \quad x \quad \ldots \quad x]$.
Instead of supplying a logit bias of $B$ for each of the five tokens,
we supply a logit bias of $B$ for token $x$,
and $B' < B$ for the other four tokens.
If $B'$ is large enough so that the other tokens will be brought into the
top-5 outputs, we will still be able to learn the logits for those tokens.
As long as $B'$ is small enough so that the model will always complete the
initial prompt $p_0$ with token $x$ (and not any other), then we will
be able to collect the logits on several prompts of the form $[p_0 \quad x \quad x \quad \ldots \quad x]$.

\textbf{Analysis.}
It is easy to see that the query cost of this attack is 
$1 \over {4m}$, where $m$ is the expansion factor.
Further, since each query requires $1+m$ tokens, the token cost is $1+m \over 4m$.
(Or, $1+m+\Delta$ if the API has an overhead of $\Delta$ tokens.)
Note that if $m=1$, i.e., there is no expansion, this attack reduces
to our first attack and the analysis similarly gives 
a query cost of $\sfrac{1}{4}$ and a token cost of $\sfrac{1}{2}$.

\begin{table*}
\centering
\begin{threeparttable}[]
\small
\centering
\caption{Attack success rate on five different black-box models}
\label{tab:model_comparison}
\begin{tabular}{@{} l rrr ccr @{}}
\toprule
& \multicolumn{3}{c}{Dimension Extraction} & \multicolumn{3}{c}{Weight Matrix Extraction} \\
\cmidrule(lr){2-2} \cmidrule(lr){2-4} \cmidrule(lr){5-7}
Model & Size & \# Queries & Cost (USD) & RMS & \# Queries & Cost (USD) \\
\midrule
OpenAI \ada{}            & \makebox[0.8em][l]{}$1024$\,\checkmark\makebox[0.8em][l]{} & $<2 \cdot 10^6$ & \$1 
& $5 \cdot 10^{-4}$ & $<2 \cdot 10^7$\makebox[1.1em][l]{} & \$4\makebox[0.8em][l]{}  \\
OpenAI \babbage{}        & \makebox[0.8em][l]{}$2048$\,\checkmark\makebox[0.8em][l]{} & $<4 \cdot 10^6$ & \$2
& $7 \cdot 10^{-4}$ & $<4 \cdot 10^7$\makebox[1.1em][l]{} & \$12\makebox[0.8em][l]{} \\
OpenAI \babbagetwo{}   & \makebox[0.8em][l]{}$1536$\,\checkmark\makebox[0.8em][l]{} & $<4 \cdot 10^6$ & \$2
& $^\dagger$ & $<4 \cdot 10^6$ \makebox[0.9em][l]{$^{\dagger+}$} & \$12\makebox[0.8em][l]{} \\
OpenAI \gptturboinstruct{}  & \makebox[1.7em][l]{}$^*$\,\,\,\checkmark\makebox[0.8em][l]{} & $<4 \cdot 10^7$ & \$200 
& $^\dagger$ & $<4 \cdot 10^8$ \makebox[0.9em][l]{$^{\dagger+}$} & \$2,000\makebox[0.8em][l]{$^{\dagger+}$} \\
OpenAI \gptturbonov{}      & \makebox[1.7em][l]{}$^*$\,\,\,\checkmark\makebox[0.8em][l]{} & $<4 \cdot 10^7$ & \$800 
& $^\dagger$ & $<4 \cdot 10^8$ \makebox[0.9em][l]{$^{\dagger+}$} & \$8,000\makebox[0.8em][l]{$^{\dagger+}$} \\
\bottomrule
\end{tabular}

    \begin{tablenotes}
    \footnotesize
    \item[\checkmark] Extracted attack size was exactly correct; confirmed in discussion with OpenAI.
    \item[$*$] As part of our responsible disclosure, OpenAI has asked that we do not publish this number.
    \item[$\dagger$] Attack not implemented to preserve security of the weights.
    \item[$+$] Estimated cost of attack given the size of the model and estimated scaling ratio.
  \end{tablenotes}
\end{threeparttable}
\end{table*}

\begin{table}
\small
\centering
\caption{Average error at recovering the logit vector for 
each of the logit-estimation attacks we develop. Our highest precision, and most efficient attack,
recovers logits nearly perfectly; other attacks approach this level of precision but at a higher
query cost.}
\label{tab:model_comparison_logit_estimation}
\vspace{5pt}
\setlength{\tabcolsep}{3pt}
\begin{tabular}{@{} l @{\hskip -5pt}c SS @{}}
\toprule

Attack & Logprobs & {Bits of precision} & {Queries per logit} \\
\midrule
logprob-4 (\S\ref{sec:topk-extraction-main}) & top-5 & 23.0 & 0.25 \\
logprob-5 (\S\ref{sec:linear_reconstruction}) & top-5 & 11.5 & 0.64 \\
logprob-1 (\S\ref{sec:binarized}) & top-1 & 6.1 & 1.0 \\
binary search (\S\ref{sec:logprobfree:basic}) & \xmark & 7.2 & 10.0 \\
hyperrectangle (\S\ref{sec:logprobfree:hyperrectangle-center}) & \xmark & 15.7 & 5.4 \\
one-of-n (\S\ref{sec:logprobfree:hyperrectangle-fancy}) & \xmark & 18.0 & 3.7 \\

\bottomrule
\end{tabular}
\end{table}

\subsection{Extraction Attack for top-1 Binary Logit Bias APIs}
\label{sec:binarized}
In light of our attacks, it is conceivable that model providers introduce  restrictions on the above API. We now demonstrate that an attack is possible even if the API only returns the top logprob ($K=1$ in the API from Section \ref{Sec:logitbiasAPI}), and the logit bias is constrained to only take one of two values. 

\textbf{API.} We place two following further restrictions on the logit bias API (Section \ref{Sec:logitbiasAPI}):
first, we set $K=1$, and only see the most likely token's logprob;
and second, each logit bias entry $\logitbias$ is constrained to be in $\{-1, 0\}$.
These constraints would completely prevent the attacks from the prior section. We believe this constraint is significantly tighter than any
practical implementation would define.

\textbf{Method.}
At first it may seem impossible to be able to learn any information
about a token $t$ if it is not already the most likely token.
However, note that if we query the model twice, once without any logit
bias, and once with a logit bias of $-1$ for token $t$, 
then the top token will be \emph{slightly} more likely with a bias of $-1$,
with exactly how slight depending on the \emph{value} of token $t$'s logprob.
Specifically, in Appendix~\ref{sec:binarized:proof} we show the logprob equals
$(\sfrac{1}{e} - 1)^{-1}(\exp(y_{\text{top}} - y'_{\text{top}}) - 1)$
where 
$y_{\text{top}}$ and $y'_{\text{top}}$ are the logprobs of the most likely
token when querying with logit bias of $0$ and $-1$.

\textbf{Analysis.}
This attack requires $1$ query and token per logprob extracted.
However, as we will show in the evaluation, this attack is much
less numerically stable than the previously-discussed attacks,
and so may require more queries to reach the same level of accuracy.

\section{Logprob-free attacks}

Due to space constraints, in Appendix~\ref{sec:logprob-free-appendix}, we show we can still extract logits \emph{without logprob access}, although with a higher cost. 

Intuitively, even without logprobs (as long as we still have logit bias)
it is possible to perform binary search to increase and decrease the logits
for every token until increasing any token by epsilon will make it the most
likely.
At this point, the logit bias vector corresponds directly to the (relative)
logits of each token relative to every other.

By performing the binary search one token at a time, we can develop an effective
(but inefficient) attack that requires $N\log({B \over \epsilon})$ where $N$ is the
number of logits, $B$ is an upper bound on the gap between any two logits and $\epsilon$ is the desired tolerance.

An improved attack is possible by noticing that it is possible to perform binary
search on multiple tokens in parallel.
Because the adversary gets to view the arg-max sampled token, by modifying multiple
tokens at the same time we can learn information faster and therefore improve attack
efficiency.

\section{Evaluation} 

We now study the efficacy of our practical
stealing attack.

\subsection{Logit Validation}

We begin by validating that the attacks developed in the prior sections
can effectively recover the full logit vector given a limited query interface.
In Table~\ref{tab:model_comparison_logit_estimation}
we report the average number of bits of agreement between the
true logit vector and the recovered logit vector, as well as
the (amortized) number of queries required to recover one full logit vector.

Generally, attacks that operate under stronger threat models have higher precision.
But theoretical improvements are not always practical: the 
theoretically stronger attack from \S\ref{sec:linear_reconstruction} that learns 5 logprobs
per query in practice requires more queries and recovers
logits with lower fidelity.
This is because this attack is numerically unstable: it requires a potentially
ill-conditioned matrix, and therefore can require re-querying the API after
adjusting the logit bias.
Our strongest logprob-free attack is highly efficient, and recovers $18$ bits of precision
at just $3.7$ queries per logit.
In Appendix~\ref{sec:howfaroptimal} we theoretically analyze how far this is from optimal,
and find it is within a factor of two.

\subsection{Stealing Parts of Production Models}
We now investigate our ability to steal production language
models, focusing on five of OpenAI's models available on 1 January 2024:
\ada{}, \babbage{}, \babbagetwo{}, \gptturboinstruct{}, and \gptturbonov{}.
We selected these models because these were the only production
models which
were able to receive advance permission to attempt an extraction
attack; we are exceptionally grateful to OpenAI for allowing
us to perform this research using their models.

Given the results from the prior section, we chose to implement the
improved 4-logprob attack (Section~\ref{sec:topk-extraction-main}) because it is both the most query efficient
attack and also the most precise attack.
Switching to a different attack algorithm would increase our total
experiment cost significantly, and so we do not perform these
ablation studies.

Both our hidden-dimension-stealing and entire-layer-stealing attack
worked for all five of these models.
The size we recover from the model perfectly matches the actual
size of the original model, as confirmed by OpenAI.
For the first three models, we report in Table~\ref{tab:model_comparison} the size we
recover because (1) the sizes of these models was never previously confirmed, but (2) they  have now been deprecated and so disclosing the size is not harmful.
In discussions with OpenAI, we decided
to withhold disclosure of the size of \gptturbo{} models,
but we confirmed with them that the number our attack reported was accurate.

When running the full layer-stealing attack, we confirmed that our extracted
weights are nearly identical to the actual weights, with error $<7 \cdot 10^{-4}$, up to an $h \times h$
matrix product as discussed previously.
Table~\ref{tab:model_comparison} reports the RMS between our extracted weight matrix and the 
actual model weights, after ``aligning'' the two by an $h \times h$ transform.

\section{Defenses}
\label{sec:defense}

It would be possible to prevent or mitigate this attack in a number of different ways,
albeit with loss of functionality.

\subsection{Prevention}

\textbf{Remove logit bias.~}
Perhaps the simplest defense would be to outright remove the logit bias parameter from the API.
Unfortunately, there are several legitimate use cases of this parameter.
For example, several works use logit bias in order to perform controlled or constrained generation \citep{jiang2023active,yang2021fudge}, to shift generation and mimic fine-tuning the model~\citep{liu2024tuning,mitchell2024emulator}, 
or other reasons \citep{ren2023selfevaluation,lee2022game}.
%

\textbf{Replace logit bias with a block-list.~}
Instead of offering a logit bias, model developers could replace it with a 
block-list of tokens the model is prohibited from emitting. 
This would support (some) of the functionality discussed in the prior section,
but would still prevent our attack.

\textbf{Architectural changes.~}
Instead of modifying the API, we could instead make changes to the model.
Our attack only works because the hidden dimension $h$ is less than
the output dimension $l$.
%
%
This suggests a natural architectural defense:
split the final layer into two layers, one that goes from $h \to t$ and then $t \to l$
where $t > l$ and a nonlinearity was placed in between.
This is not very efficient though, as the last linear layer
is large (quadratic in the vocabulary size).


\textbf{Post-hoc altering the architecture.~} We can also modify the hidden dimension $h$ for the final layer after the model is trained. In particular, we can expand the dimensionality of $\Et$ by concatenating extra weight vectors that are orthogonal to the original matrix. We set the singular values for these weights to be small enough to not materially affect the model's predictions, while also being large enough to look realistic. Then, during the model's forward pass, we concatenate a vector of random Gaussian noise to the final hidden vector $g_\theta(\x)$ before multiplying by $\Et$. Figure~\ref{fig:gpt2_with_spoofing} shows an example of this, where we expand GPT-2 small to appear as if it was 1024 dimensional instead of 768 dimensions. This misleads the adversary into thinking that the model is wider than it actually is.

\subsection{Mitigations}\label{sec:mitigations}

\textbf{Logit bias XOR logprobs.~}
Our attack is $10\times$ cheaper when an adversary can supply both a logit bias and also
view output logprobs.
This suggests a natural mitigation:
prohibit queries to the API that make use of \emph{both} logit bias
and logprobs at the same time.
This type of defense is common in both the security and machine learning community:
for example, in 2023 OpenAI removed the ability to combine both
\nicett{echo} and \nicett{logprobs},
but with either alone being allowed;
this defense would behave similarly.

\textbf{Noise addition.~}
By adding a sufficient amount of noise to the output logits of any given query,
it would be possible to prevent our attack.
However, logit-noise has the potential to make models less useful.
We perform some preliminary experiments on this direction in Appendix~\ref{app:quantnoise}.

\textbf{Rate limits on logit bias.~}
Our attack requires that we are able to learn at least $h$ logit
values for each prompt $p$.
One defense would be to allow logit-bias queries to the model,
but only allow $T = \tilde{h}/5$ logit bias queries for
any given prompt $p$ to prevent an adversary from learning
if a model has hidden dimension $\tilde{h}$ or smaller.

Unfortunately this has several significant drawbacks:
the threshold has to be independent of $h$ (or learning the threshold would reveal $h$);
the system would need to maintain state of all user queries to the API;
and preventing Sybil attacks requires a global pool of user queries, which can present significant privacy risks~\citep{debenedetti2023privacy}.

\textbf{Detect malicious queries.~}
Instead of preventing any queries that might leak model weights, an
alternate strategy could be to implement standard anti-abuse
tools to \emph{detect} any patterns of malicious behavior.
Several proposals of this form exist for prior machine learning attacks,
including model stealing \cite{juuti2019prada,pal2021stateful}
and adversarial examples \cite{chen2020stateful}.
Unfortunately, these defenses are often vulnerable to attack
\citep{feng2023stateful},
and so any mitigation here would need to improve
on the state-of-the-art to be truly robust.

\section{Future Work}

We are motivated to study this problem not because we expect to be able to
steal an entire production transformer model bit-for-bit,
but because we hope to conclusively demonstrate that model stealing attacks
are not just of academic concern but can be practically applied to the
largest production models deployed today.
We see a number of potential directions for improving on this attack.

\textbf{Breaking symmetry with quantized weights.}
Large production models are typically stored ``quantized'',
where each weight is represented in just 4 or 8 bits.
In principle, this quantization could allow an adversary to recover
a nearly bit-for-bit copy of the matrix $\Et$: while there exist
an infinite number of matrices $\Et \cdot \mathbf{G}$, only one will be discretized properly.
Unfortunately, this integer-constrained problem is NP-hard in general (similar problems are the foundation for an entire class of public key
cryptosystems).
But this need not imply that the problem is hard on all instances.

\textbf{Extending this attack beyond a single layer.}
Our attack recovers a single layer of a transformer.
We see no obvious methodology to extend it beyond just
a single layer, due to the non-linearity of the models.
But we invite further research in this area.

\textbf{Removing the logit bias assumption.}
All our attacks require the ability to pass a logit bias.
Model providers including Google and OpenAI provided this 
capability when we began the writing of this paper, but this could change.
(Indeed, it already has, as model providers
begin implementing defenses to prevent this attack.)
Other API parameters could give alternative avenues for learning logit information.
For example, unconstrained \nicett{temperature} and \nicett{top-k} parameters could also
leak logit values through a series of queries.
In the long run, completely hiding the logit information might be challenging due
both to public demand for the feature, and ability of adversaries to infer this
information through other means.

\textbf{Exploiting the stolen weights.} Recovering a model's embedding projection layer might improve other attacks against that model. 
Alternatively, an attacker could infer details about a provider's \emph{finetuning} API by observing changes (or the absence thereof) in the last layer. In this paper, we focus primarily on the model extraction problem and leave exploring downstream attacks to future work.

\textbf{Practical stealing of other model information.}
Existing high-fidelity model stealing attacks are ``all-or-nothing'' attacks
that recover entire models, but only apply to small ReLU networks.
We show that stealing partial information can be much more practical, even for state-of-the-art models.
Future work may find that practical attacks can steal many more bits of information about current proprietary models.

\section{Conclusion}

As the field of machine learning matures, and models transition from research artifacts
to production tools used by millions, the field of adversarial machine
learning must also adapt.
While it is certainly useful to understand the potential applicability of model stealing to
three-layer 100-neuron ReLU-only fully-connected networks, 
at some point it becomes important to understand to what extent attacks can be
actually applied to the largest production models.

This paper takes one step in that direction.
We give an existence proof that it is possible to steal one layer of a production language model.
While there appear to be no immediate
practical consequences of learning this layer, it represents the first time
that \emph{any} precise information about a deployed transformer model has been stolen. 
Two immediate open questions are (1) how hazardous these practical stealing attacks are and (2) whether they pose a greater threat to developers and the security of their models than black-box access already does via distillation or other approximate stealing attacks.

Our attack also highlights how small design decisions influence the overall security of a system.
Our attack works because of the seemingly-innocuous \emph{logit-bias} and \emph{logprobs} parameters
made available by the largest machine learning service providers, including OpenAI and Google---although both have now implemented mitigations to prevent this attack or make it more expensive.
Practitioners should strive to understand how system-level design
decisions impact the safety and security of the full product.

Overall, we hope our paper serves to further motivate the study of practical attacks
on machine learning models, in order to ultimately develop safer and more reliable systems.

\section*{Impact Statement}

This paper is the most recent in a line of work that demonstrates successful attacks on production models.
As such, we take several steps to mitigate the near-term potential harms of this research.
As discussed throughout the paper, we have worked closely with all affected products
to ensure that mitigations are in place before disclosing this work.
We have additionally sent advance copies of this paper to all \emph{potentially}
affected parties, even if we were unable to precisely verify our attack.

Long-term, we believe that openly discussing vulnerabilities that have practical
impact is an important strategy for ensuring safe machine learning.
This vulnerability exists whether or not we report on it.
Especially for attacks that are simple to identify
(as evidenced by the concurrent work of \citet{finlayson2024logits} that discovered this same vulnerability),
malicious actors are also likely to discover the same vulnerability whether or
not we report on it.
By documenting it early, we can ensure future systems remain secure.


\section*{Acknowledgements}

We are grateful to Andreas Terzis and the anonymous reviewers for comments on early
drafts of this paper.
%
%
We are grateful to Joshua Achiam for helping to write the code for post-hoc modifying the model architecture. 
We are also grateful to OpenAI for allowing us to attempt our
extraction attack on their production models.

\appendix

\bibliography{refs}

\begin{thebibliography}{41}
\providecommand{\natexlab}[1]{#1}
\providecommand{\url}[1]{\texttt{#1}}
\expandafter\ifx\csname urlstyle\endcsname\relax
  \providecommand{\doi}[1]{doi: #1}\else
  \providecommand{\doi}{doi: \begingroup \urlstyle{rm}\Url}\fi

\bibitem[Anil et~al.(2023)]{anil2023palm}
Anil, R. et~al.
\newblock {PaLM 2 Technical Report}, 2023.

\bibitem[Ba et~al.(2016)Ba, Kiros, and Hinton]{ba2016layer}
Ba, J.~L., Kiros, J.~R., and Hinton, G.~E.
\newblock Layer normalization.
\newblock \emph{arXiv preprint arXiv:1607.06450}, 2016.

\bibitem[Biderman(2024)]{stellabidermangooglesheet}
Biderman, S.
\newblock Common {LLM} settings, 2024.
\newblock URL \url{https://rb.gy/2afqlw}.
\newblock Accessed February 1, 2024.

\bibitem[Biderman et~al.(2023)Biderman, Schoelkopf, Anthony, Bradley,
  O’Brien, Hallahan, Khan, Purohit, Prashanth, Raff,
  et~al.]{biderman2023pythia}
Biderman, S., Schoelkopf, H., Anthony, Q.~G., Bradley, H., O’Brien, K.,
  Hallahan, E., Khan, M.~A., Purohit, S., Prashanth, U.~S., Raff, E., et~al.
\newblock Pythia: A suite for analyzing large language models across training
  and scaling.
\newblock In \emph{International Conference on Machine Learning}, 2023.

\bibitem[Cancedda(2024)]{spectral}
Cancedda, N.
\newblock Spectral filters, dark signals, and attention sinks, 2024.

\bibitem[Carlini et~al.(2020)Carlini, Jagielski, and
  Mironov]{carlini2020cryptanalytic}
Carlini, N., Jagielski, M., and Mironov, I.
\newblock Cryptanalytic extraction of neural network models.
\newblock In \emph{Annual International Cryptology Conference}, 2020.

\bibitem[Chen et~al.(2020)Chen, Carlini, and Wagner]{chen2020stateful}
Chen, S., Carlini, N., and Wagner, D.
\newblock Stateful detection of black-box adversarial attacks.
\newblock In \emph{Proceedings of the 1st ACM Workshop on Security and Privacy
  on Artificial Intelligence}, 2020.

\bibitem[Chiu(2024)]{openlogprobs}
Chiu, J.
\newblock openlogprobs, 2024.
\newblock URL \url{https://github.com/justinchiu/openlogprobs}.
\newblock Accessed February 1, 2024.

\bibitem[Debenedetti et~al.(2023)Debenedetti, Severi, Carlini, Choquette-Choo,
  Jagielski, Nasr, Wallace, and Tram{\`e}r]{debenedetti2023privacy}
Debenedetti, E., Severi, G., Carlini, N., Choquette-Choo, C.~A., Jagielski, M.,
  Nasr, M., Wallace, E., and Tram{\`e}r, F.
\newblock Privacy side channels in machine learning systems.
\newblock \emph{arXiv preprint arXiv:2309.05610}, 2023.

\bibitem[Dettmers et~al.(2022)Dettmers, Lewis, Shleifer, and
  Zettlemoyer]{dettmers2022optimizers}
Dettmers, T., Lewis, M., Shleifer, S., and Zettlemoyer, L.
\newblock {8-bit Optimizers via Block-wise Quantization}.
\newblock \emph{ICLR}, 2022.

\bibitem[Elhage et~al.(2021)Elhage, Nanda, Olsson, Henighan, Joseph, Mann,
  Askell, Bai, Chen, Conerly, DasSarma, Drain, Ganguli, Hatfield-Dodds,
  Hernandez, Jones, Kernion, Lovitt, Ndousse, Amodei, Brown, Clark, Kaplan,
  McCandlish, and Olah]{elhage2021mathematical}
Elhage, N., Nanda, N., Olsson, C., Henighan, T., Joseph, N., Mann, B., Askell,
  A., Bai, Y., Chen, A., Conerly, T., DasSarma, N., Drain, D., Ganguli, D.,
  Hatfield-Dodds, Z., Hernandez, D., Jones, A., Kernion, J., Lovitt, L.,
  Ndousse, K., Amodei, D., Brown, T., Clark, J., Kaplan, J., McCandlish, S.,
  and Olah, C.
\newblock A mathematical framework for transformer circuits.
\newblock 2021.
\newblock URL \url{https://transformer-circuits.pub/2021/framework/index.html}.

\bibitem[Feng et~al.(2023)Feng, Hooda, Mangaokar, Fawaz, Jha, and
  Prakash]{feng2023stateful}
Feng, R., Hooda, A., Mangaokar, N., Fawaz, K., Jha, S., and Prakash, A.
\newblock Stateful defenses for machine learning models are not yet secure
  against black-box attacks.
\newblock In \emph{Proceedings of the 2023 ACM SIGSAC Conference on Computer
  and Communications Security}, pp.\  786--800, 2023.

\bibitem[Finlayson et~al.(2024)Finlayson, Swayamdipta, and
  Ren]{finlayson2024logits}
Finlayson, M., Swayamdipta, S., and Ren, X.
\newblock Logits of api-protected llms leak proprietary information.
\newblock \emph{arXiv preprint arXiv:2403.09539}, 2024.

\bibitem[Gao(2021)]{gao2021sizes}
Gao, L.
\newblock On the sizes of {OpenAI API} models.
\newblock \url{https://blog.eleuther.ai/gpt3-model-sizes/}, 2021.

\bibitem[Google(2024)]{google:changelog-1.38.0}
Google.
\newblock Changelog 1.38.0, 2024.
\newblock URL
  \url{https://cloud.google.com/python/docs/reference/aiplatform/1.38.0/changelog}.
\newblock Accessed January 30, 2024.

\bibitem[Gurnee et~al.(2024)Gurnee, Horsley, Guo, Kheirkhah, Sun, Hathaway,
  Nanda, and Bertsimas]{wesuniversal}
Gurnee, W., Horsley, T., Guo, Z.~C., Kheirkhah, T.~R., Sun, Q., Hathaway, W.,
  Nanda, N., and Bertsimas, D.
\newblock Universal neurons in gpt2 language models, 2024.

\bibitem[Hayase et~al.(2024)Hayase, Borevkovic, Carlini, Tram{\`e}r, and
  Nasr]{hayase2024query}
Hayase, J., Borevkovic, E., Carlini, N., Tram{\`e}r, F., and Nasr, M.
\newblock Query-based adversarial prompt generation.
\newblock \emph{arXiv preprint arXiv:2402.12329}, 2024.

\bibitem[Jagielski et~al.(2020)Jagielski, Carlini, Berthelot, Kurakin, and
  Papernot]{jagielski2020high}
Jagielski, M., Carlini, N., Berthelot, D., Kurakin, A., and Papernot, N.
\newblock High accuracy and high fidelity extraction of neural networks.
\newblock In \emph{USENIX Security Symposium}, 2020.

\bibitem[Jiang et~al.(2023)Jiang, Xu, Gao, Sun, Liu, Dwivedi-Yu, Yang, Callan,
  and Neubig]{jiang2023active}
Jiang, Z., Xu, F., Gao, L., Sun, Z., Liu, Q., Dwivedi-Yu, J., Yang, Y., Callan,
  J., and Neubig, G.
\newblock Active retrieval augmented generation.
\newblock In \emph{EMNLP}, 2023.

\bibitem[Juuti et~al.(2019)Juuti, Szyller, Marchal, and Asokan]{juuti2019prada}
Juuti, M., Szyller, S., Marchal, S., and Asokan, N.
\newblock {PRADA}: protecting against {DNN} model stealing attacks.
\newblock In \emph{EuroS\&P}, 2019.

\bibitem[Lee et~al.(2022)Lee, Nachum, Yang, Lee, Freeman, Guadarrama, Fischer,
  Xu, Jang, Michalewski, and Mordatch]{lee2022game}
Lee, K.-H., Nachum, O., Yang, M.~S., Lee, L., Freeman, D., Guadarrama, S.,
  Fischer, I., Xu, W., Jang, E., Michalewski, H., and Mordatch, I.
\newblock Multi-game decision transformers.
\newblock In \emph{Advances in Neural Information Processing Systems}, 2022.

\bibitem[Liu et~al.(2024)Liu, Han, Wang, Tsvetkov, Choi, and
  Smith]{liu2024tuning}
Liu, A., Han, X., Wang, Y., Tsvetkov, Y., Choi, Y., and Smith, N.~A.
\newblock Tuning language models by proxy.
\newblock \emph{arXiv preprint arXiv:2401.08565}, 2024.

\bibitem[Milli et~al.(2019)Milli, Schmidt, Dragan, and Hardt]{milli2019model}
Milli, S., Schmidt, L., Dragan, A.~D., and Hardt, M.
\newblock Model reconstruction from model explanations.
\newblock In \emph{Proceedings of the Conference on Fairness, Accountability,
  and Transparency}, 2019.

\bibitem[Mitchell et~al.(2024)Mitchell, Rafailov, Sharma, Finn, and
  Manning]{mitchell2024emulator}
Mitchell, E., Rafailov, R., Sharma, A., Finn, C., and Manning, C.~D.
\newblock An emulator for fine-tuning large language models using small
  language models.
\newblock In \emph{ICLR}, 2024.

\bibitem[Morris et~al.(2023)Morris, Zhao, Chiu, Shmatikov, and
  Rush]{morris2023language}
Morris, J.~X., Zhao, W., Chiu, J.~T., Shmatikov, V., and Rush, A.~M.
\newblock Language model inversion.
\newblock \emph{arXiv preprint arXiv:2311.13647}, 2023.

\bibitem[OpenAI(2023)]{openai:logit-bias}
OpenAI.
\newblock Using logit bias to define token probability, 2023.
\newblock URL
  \url{https://help.openai.com/en/articles/5247780-using-logit-bias-to-define-token-probability}.
\newblock Accessed Febraury 1, 2024.

\bibitem[OpenAI(2024)]{openai:create-chat-completion}
OpenAI.
\newblock Create chat completion, 2024.
\newblock URL \url{https://platform.openai.com/docs/api-reference/chat/create}.
\newblock Accessed January 30, 2024.

\bibitem[OpenAI et~al.(2023)]{openai2023gpt4}
OpenAI et~al.
\newblock {GPT-4 Technical Report}, 2023.

\bibitem[Pal et~al.(2021)Pal, Gupta, Kanade, and Shevade]{pal2021stateful}
Pal, S., Gupta, Y., Kanade, A., and Shevade, S.
\newblock Stateful detection of model extraction attacks.
\newblock \emph{arXiv preprint arXiv:2107.05166}, 2021.

\bibitem[Radford et~al.(2019)Radford, Wu, Child, Luan, Amodei, and
  Sutskever]{radford2019gpt2}
Radford, A., Wu, J., Child, R., Luan, D., Amodei, D., and Sutskever, I.
\newblock {Language Models are Unsupervised Multitask Learners}.
\newblock Technical report, OpenAI, 2019.
\newblock URL \url{https://rb.gy/tm8qh}.

\bibitem[Rae et~al.(2022)Rae, Borgeaud, Cai, Millican, Hoffmann, Song,
  Aslanides, Henderson, Ring, Young, Rutherford, Hennigan, Menick, Cassirer,
  Powell, van~den Driessche, Hendricks, Rauh, Huang, Glaese, Welbl, Dathathri,
  Huang, Uesato, Mellor, Higgins, Creswell, McAleese, Wu, Elsen, Jayakumar,
  Buchatskaya, Budden, Sutherland, Simonyan, Paganini, Sifre, Martens, Li,
  Kuncoro, Nematzadeh, Gribovskaya, Donato, Lazaridou, Mensch, Lespiau,
  Tsimpoukelli, Grigorev, Fritz, Sottiaux, Pajarskas, Pohlen, Gong, Toyama,
  de~Masson~d'Autume, Li, Terzi, Mikulik, Babuschkin, Clark, de~Las~Casas, Guy,
  Jones, Bradbury, Johnson, Hechtman, Weidinger, Gabriel, Isaac, Lockhart,
  Osindero, Rimell, Dyer, Vinyals, Ayoub, Stanway, Bennett, Hassabis,
  Kavukcuoglu, and Irving]{gopher}
Rae, J.~W., Borgeaud, S., Cai, T., Millican, K., Hoffmann, J., Song, F.,
  Aslanides, J., Henderson, S., Ring, R., Young, S., Rutherford, E., Hennigan,
  T., Menick, J., Cassirer, A., Powell, R., van~den Driessche, G., Hendricks,
  L.~A., Rauh, M., Huang, P.-S., Glaese, A., Welbl, J., Dathathri, S., Huang,
  S., Uesato, J., Mellor, J., Higgins, I., Creswell, A., McAleese, N., Wu, A.,
  Elsen, E., Jayakumar, S., Buchatskaya, E., Budden, D., Sutherland, E.,
  Simonyan, K., Paganini, M., Sifre, L., Martens, L., Li, X.~L., Kuncoro, A.,
  Nematzadeh, A., Gribovskaya, E., Donato, D., Lazaridou, A., Mensch, A.,
  Lespiau, J.-B., Tsimpoukelli, M., Grigorev, N., Fritz, D., Sottiaux, T.,
  Pajarskas, M., Pohlen, T., Gong, Z., Toyama, D., de~Masson~d'Autume, C., Li,
  Y., Terzi, T., Mikulik, V., Babuschkin, I., Clark, A., de~Las~Casas, D., Guy,
  A., Jones, C., Bradbury, J., Johnson, M., Hechtman, B., Weidinger, L.,
  Gabriel, I., Isaac, W., Lockhart, E., Osindero, S., Rimell, L., Dyer, C.,
  Vinyals, O., Ayoub, K., Stanway, J., Bennett, L., Hassabis, D., Kavukcuoglu,
  K., and Irving, G.
\newblock Scaling language models: Methods, analysis and insights from training
  gopher, 2022.

\bibitem[Ren et~al.(2023)Ren, Zhao, Vu, Liu, and
  Lakshminarayanan]{ren2023selfevaluation}
Ren, J., Zhao, Y., Vu, T., Liu, P.~J., and Lakshminarayanan, B.
\newblock Self-evaluation improves selective generation in large language
  models.
\newblock \emph{arXiv preprint arXiv:2312.09300}, 2023.

\bibitem[Rolnick \& Kording(2020)Rolnick and Kording]{rolnick2020reverse}
Rolnick, D. and Kording, K.
\newblock Reverse-engineering deep relu networks.
\newblock In \emph{International Conference on Machine Learning}, 2020.

\bibitem[Shamir et~al.(2023)Shamir, Canales-Martinez, Hambitzer, Chavez-Saab,
  Rodrigez-Henriquez, and Satpute]{shamir2023polynomial}
Shamir, A., Canales-Martinez, I., Hambitzer, A., Chavez-Saab, J.,
  Rodrigez-Henriquez, F., and Satpute, N.
\newblock Polynomial time cryptanalytic extraction of neural network models.
\newblock \emph{arXiv preprint arXiv:2310.08708}, 2023.

\bibitem[Touvron et~al.(2023)Touvron, Lavril, Izacard, Martinet, Lachaux,
  Lacroix, Rozi{\`e}re, Goyal, Hambro, Azhar, et~al.]{touvron2023llama}
Touvron, H., Lavril, T., Izacard, G., Martinet, X., Lachaux, M.-A., Lacroix,
  T., Rozi{\`e}re, B., Goyal, N., Hambro, E., Azhar, F., et~al.
\newblock {LLaMA}: Open and efficient foundation language models.
\newblock \emph{arXiv preprint arXiv:2302.13971}, 2023.

\bibitem[Tram{\`e}r et~al.(2016)Tram{\`e}r, Zhang, Juels, Reiter, and
  Ristenpart]{tramer2016stealing}
Tram{\`e}r, F., Zhang, F., Juels, A., Reiter, M.~K., and Ristenpart, T.
\newblock Stealing machine learning models via prediction {APIs}.
\newblock In \emph{USENIX Security Symposium}, 2016.

\bibitem[Veit et~al.(2016)Veit, Wilber, and Belongie]{resnet_shallow_path}
Veit, A., Wilber, M.~J., and Belongie, S.~J.
\newblock Residual networks behave like ensembles of relatively shallow
  networks.
\newblock In \emph{Advances in Neural Information Processing Systems}, pp.\
  550--558, 2016.

\bibitem[Wei et~al.(2020)Wei, Zhang, Zhou, Li, and Al~Faruque]{wei2020leaky}
Wei, J., Zhang, Y., Zhou, Z., Li, Z., and Al~Faruque, M.~A.
\newblock Leaky {DNN}: Stealing deep-learning model secret with {GPU}
  context-switching side-channel.
\newblock In \emph{IEEE/IFIP International Conference on Dependable Systems and
  Networks (DSN)}, 2020.

\bibitem[Yang \& Klein(2021)Yang and Klein]{yang2021fudge}
Yang, K. and Klein, D.
\newblock {FUDGE}: Controlled text generation with future discriminators.
\newblock In Toutanova, K., Rumshisky, A., Zettlemoyer, L., Hakkani-Tur, D.,
  Beltagy, I., Bethard, S., Cotterell, R., Chakraborty, T., and Zhou, Y.
  (eds.), \emph{ACL}, 2021.

\bibitem[Zanella-Beguelin et~al.(2021)Zanella-Beguelin, Tople, Paverd, and
  K{\"o}pf]{zanella2021grey}
Zanella-Beguelin, S., Tople, S., Paverd, A., and K{\"o}pf, B.
\newblock Grey-box extraction of natural language models.
\newblock In \emph{International Conference on Machine Learning}, pp.\
  12278--12286. PMLR, 2021.

\bibitem[Zhang \& Sennrich(2019)Zhang and Sennrich]{zhang2019root}
Zhang, B. and Sennrich, R.
\newblock Root mean square layer normalization.
\newblock \emph{NeurIPS}, 2019.

\end{thebibliography}
\bibliographystyle{icml2024}

\begin{appendix}
\onecolumn


\section{What's Going On With GPT-2 Small?}
\label{sec:whygpt2}

Our attack nearly perfectly extracts the model size of all models---except for GPT-2 Small
where our extracted size of 757 is off by 11 from the correct 768.
Why is this?

In Figure~\ref{fig:gpt2small} we directly inspect this model's final hidden activation vector
across $10,000$ different model queries and perform SVD of the resulting activation matrix.
We see that despite GPT-2 actually having 768 potential hidden neurons,
there are only $757$ different activation directions. 
Thus, while this model is \emph{technically} a 768 dimensional model, 
in practice it behaves as if it was a 757 (i.e, the rank of the embedding matrix is 757) dimensional model,
and our attack has recovered this effective size.

However, when running the model in higher float64 precision, we find that indeed all dimensions are used, but that the smallest dozen or so singular values are much smaller than the other singular values, an observation made by concurrent work \citep{spectral}.

\begin{figure}[h]
  \begin{minipage}{0.5\textwidth}
    \centering
    \includegraphics[width=\linewidth]{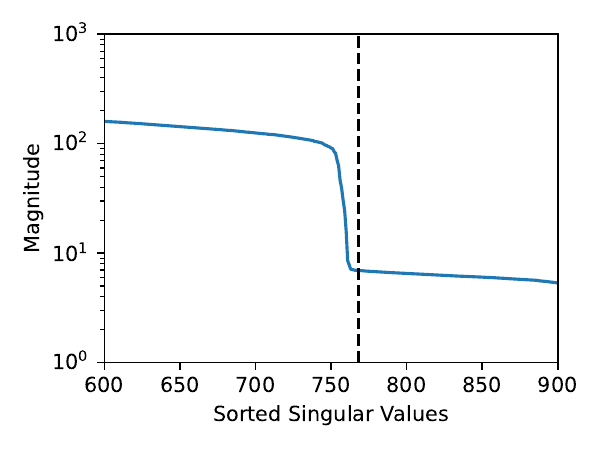}
    \captionof{subfigure}{Singular values of GPT-2 Small (default bfloat16 precision)}
  \end{minipage}
  \hfill
  \begin{minipage}{0.5\textwidth}
    \centering
    \vspace{-.1cm}
    \includegraphics[width=0.9\linewidth]{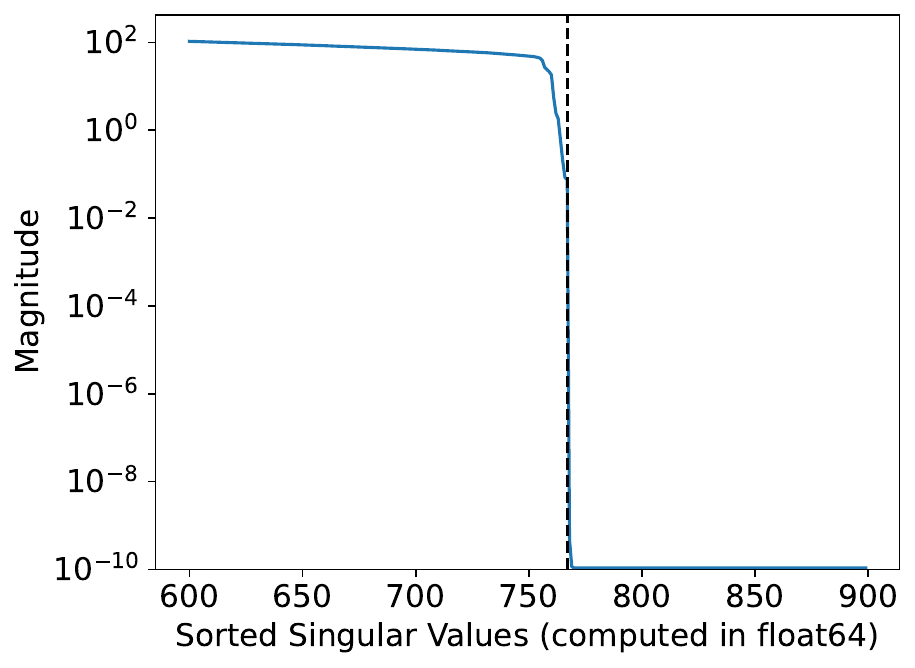}
    \captionof{subfigure}{Singular values of GPT-2 Small (higher float64 precision)
    \label{subfig:gpt2smallfloat64}}
  \end{minipage}
  \hfill
  \caption{Singular values of final hidden activations of GPT-2 Small.}
  \label{fig:gpt2small}
\end{figure}

\section{Accounting for Normalization Layers} 
\label{app:proofs}

\subsection{LayerNorm Does Not Affect Our Rank $h$ Assumption}
\label{subapp:normalization_layers}

Almost all LLMs that have publicly available architecture details use LayerNorm \citep{ba2016layer} or RMSNorm \citep{zhang2019root} just before applying the output projection $\Et$ \citep{stellabidermangooglesheet}. LayerNorm begins with a centering step, which projects its input onto a $(h-1)$-dimensional subspace (and RMSNorm does not). In theory, this could break our assumption that the rank of the matrix with columns $g_\theta\br{\x_i}$ ($i=1, ..., n$) has rank $h$ (\Cref{lem:h_recovery}). In practice, all LLMs we surveyed \citep{stellabidermangooglesheet} enabled the LayerNorm bias, which means the matrices had full rank $h$ (besides GPT-2 Small: see \Cref{sec:whygpt2}).

\subsection{Stealing Architectural Details About Normalization Layers}
\label{subsec:stealing-architecture}

\subsubsection{Theory}
\label{subsubsec:stealing-theory}

The difference between LayerNorm and RMSNorm (\Cref{subapp:normalization_layers}) could enable attackers to deduce whether models used LayerNorm or RMSNorm.  If an attacker recovered an initial logit-vector API query response $\api\br{\x_0}$, then they could apply \Cref{lem:h_recovery} to $\api\br{\x_1}-\api\br{\x_0}, \ldots, \api\br{\x_{n}}-\api\br{\x_0}$.\footnote{Throughout this appendix section, we assume the sum of logit outputs is always 0. We can calculate centered logits from logprobs by subtracting the mean logits across the vocab dimension.} From the description of the API at the top of  
\Cref{sec:attack_warmup}, 
it follows that $\api\br{\x_i} - \api\br{\x_0} = \Et ( g_\theta\br{\x_i} - g_\theta\br{\x_0} )$. This subtraction of $g$ terms occurs immediately after LayerNorm, so cancels the LayerNorm bias term. Hence, if we apply the \Cref{lem:h_recovery} attack with this subtraction modification to a model using LayerNorm, then the resultant `$h$' output will be smaller by 1 (due to \Cref{subapp:normalization_layers}). This would imply the model used LayerNorm rather than RMSNorm, because RMSNorm does not project onto a smaller subspace and so would not have a decrease in `$h$' value if we were to use this subtraction trick.

\subsubsection{Results}
\label{subsubsec:stealing-results}

To confirm that the method from \Cref{subsubsec:stealing-theory} works, we test whether we can detect whether the GPT-2, Pythia and LLAMA architectures use LayerNorm or RMSNorm from their logit outputs alone. We found that the technique required two adjustments before it worked on models with lower than 32-bit precision (it always worked with 32-bit precision). i) We do not subtract $\api\br{\x_0}$ from logits queries, but instead subtract the mean logits over all queries, i.e. $\frac1n \sum_{i=1}^n \api\br{\x_i}$. Since the average of several points in a common affine subspace still lie on that affine subspace, this doesn't change the conclusions from \Cref{subsubsec:stealing-theory}. ii) We additionally found it helped to calculate this mean in lower precision, before casting to 64-bit precision to calculate the compact SVD.

The results are in \Cref{fig:stealing-norm}. We plot the singular value magnitudes (as in \Cref{fig:why_svd_works}) and show that \textbf{there is a drop in the $h$th singular value for the architectures using LayerNorm, but not for architecture using RMSNorm}:

\begin{figure}[h!]
  \begin{minipage}{0.3\textwidth}
    \centering
    \vspace{-.4cm}
    \includegraphics[width=1.15\linewidth]{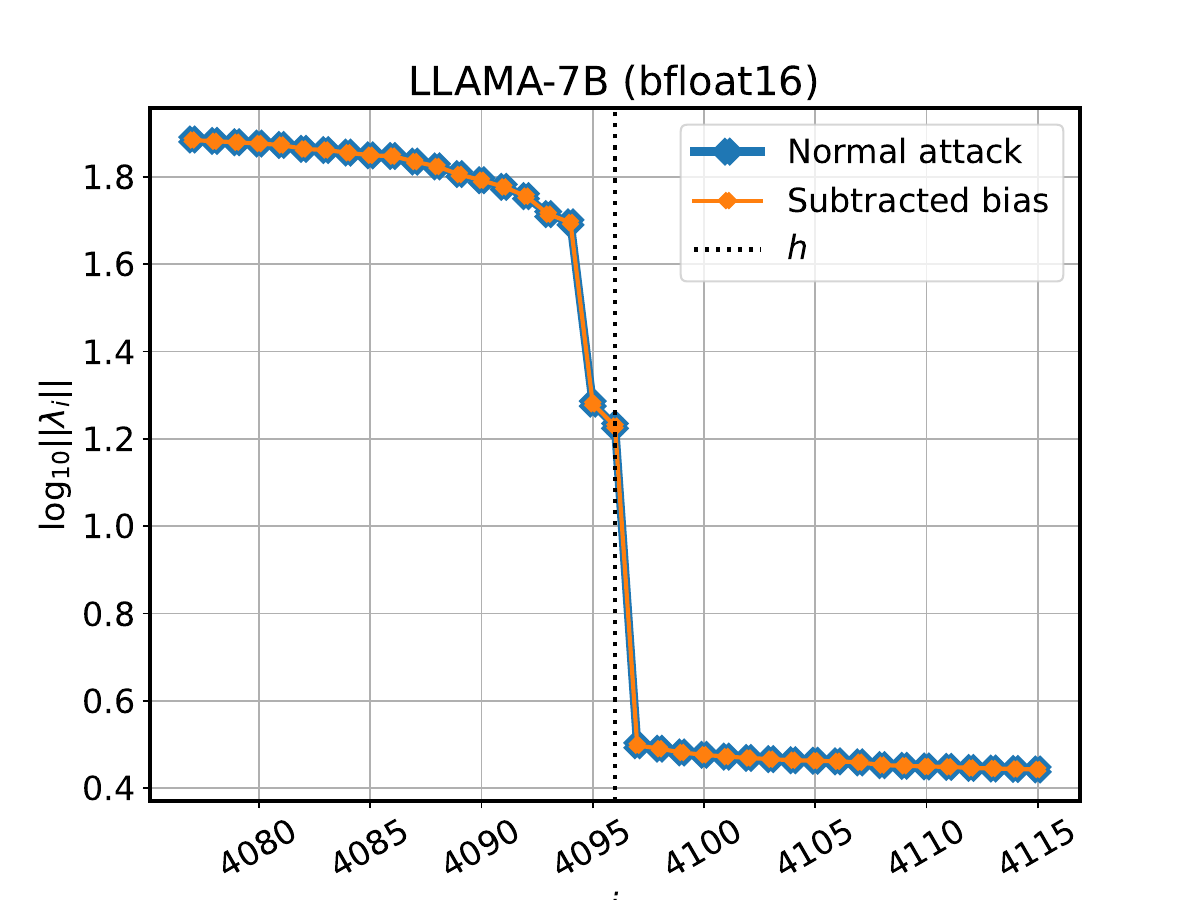}
    \footnotesize{(a) LLAMA-7B (\textbf{RMSNorm}).}
    \label{subfig:llama_abs}
  \end{minipage}
  \hfill
  \begin{minipage}{0.3\textwidth}
    \centering
    \includegraphics[width=1.15\linewidth]{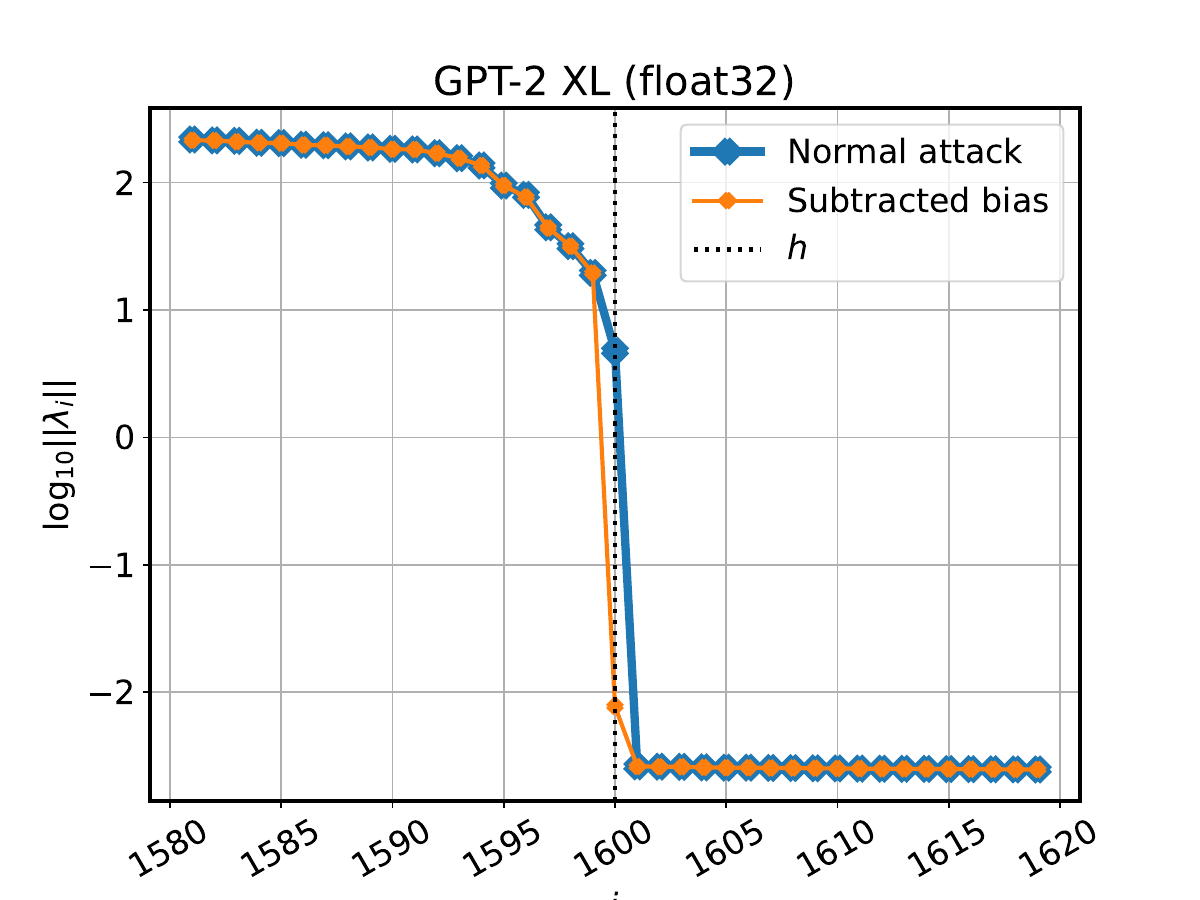}
    \footnotesize{(b) GPT-2 XL (\textbf{LayerNorm}).} 
    \label{subfig:gpt2_abs}
  \end{minipage}
  \hfill
  \begin{minipage}{0.3\textwidth}
    \centering
    \vspace{-.1cm}
    \includegraphics[width=1.15\linewidth]{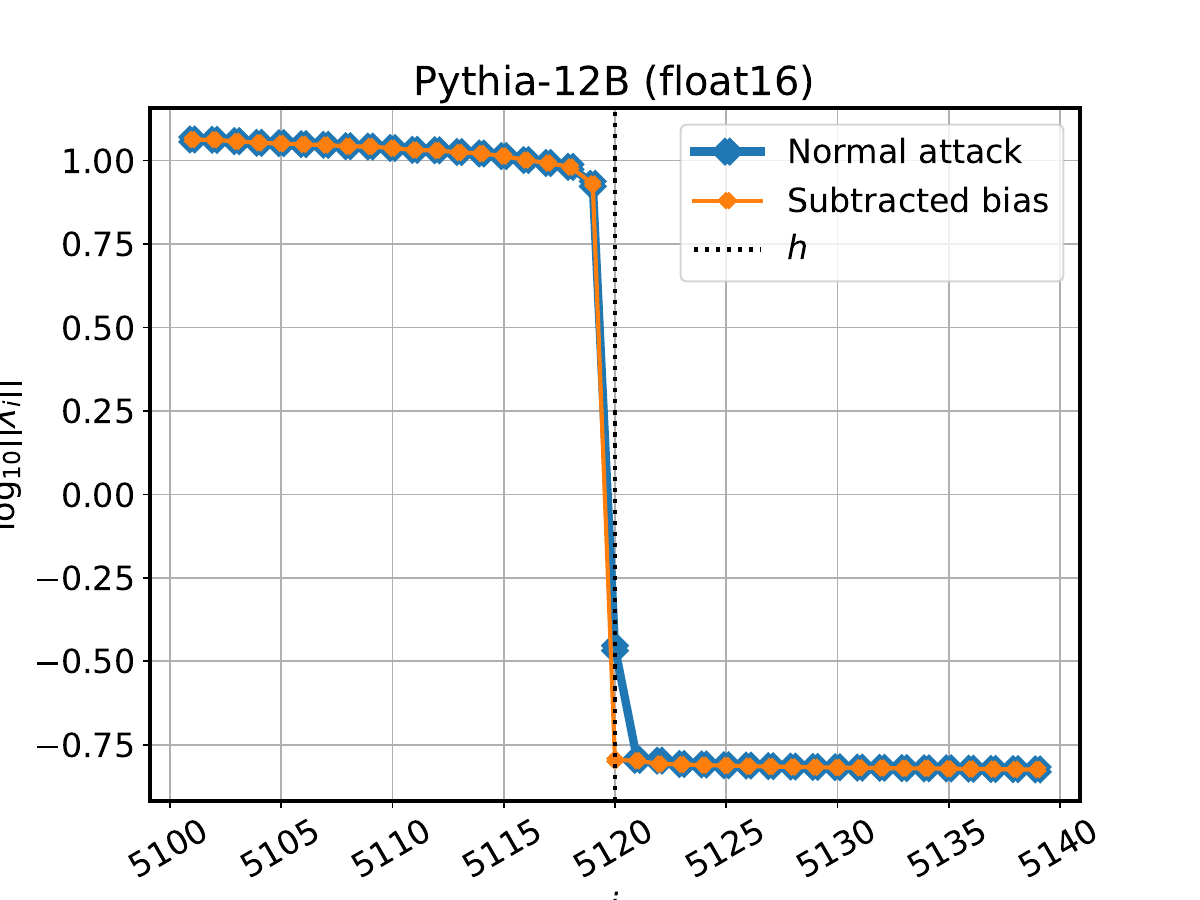}
    \footnotesize{(c) Pythia-12B (\textbf{LayerNorm}).}
    \label{subfig:pythia_abs}
  \end{minipage}
  \caption{Detecting whether models use \textbf{LayerNorm} or \textbf{RMSNorm} by singular value magnitudes.}
  \label{fig:stealing-norm}
\end{figure}

Is this attack practical for real models? We perform the same attack on the logprobs we obtained for \texttt{ada} and \texttt{babbage}.\footnote{Unfortunately, we deleted the logprobs for GPT-3.5 models before we created this attack due to security constraints.} We see in \Cref{fig:scaling-stealing-norm}a-b that indeed the drop in the $h$th singular values occurs for these two models that use LayerNorm (GPT-3's architecture was almost entirely inherited from GPT-2):

\begin{figure}[h!]
  \begin{minipage}{0.3\textwidth}
    \centering
    \vspace{-.4cm}
    \includegraphics[width=1.15\linewidth]{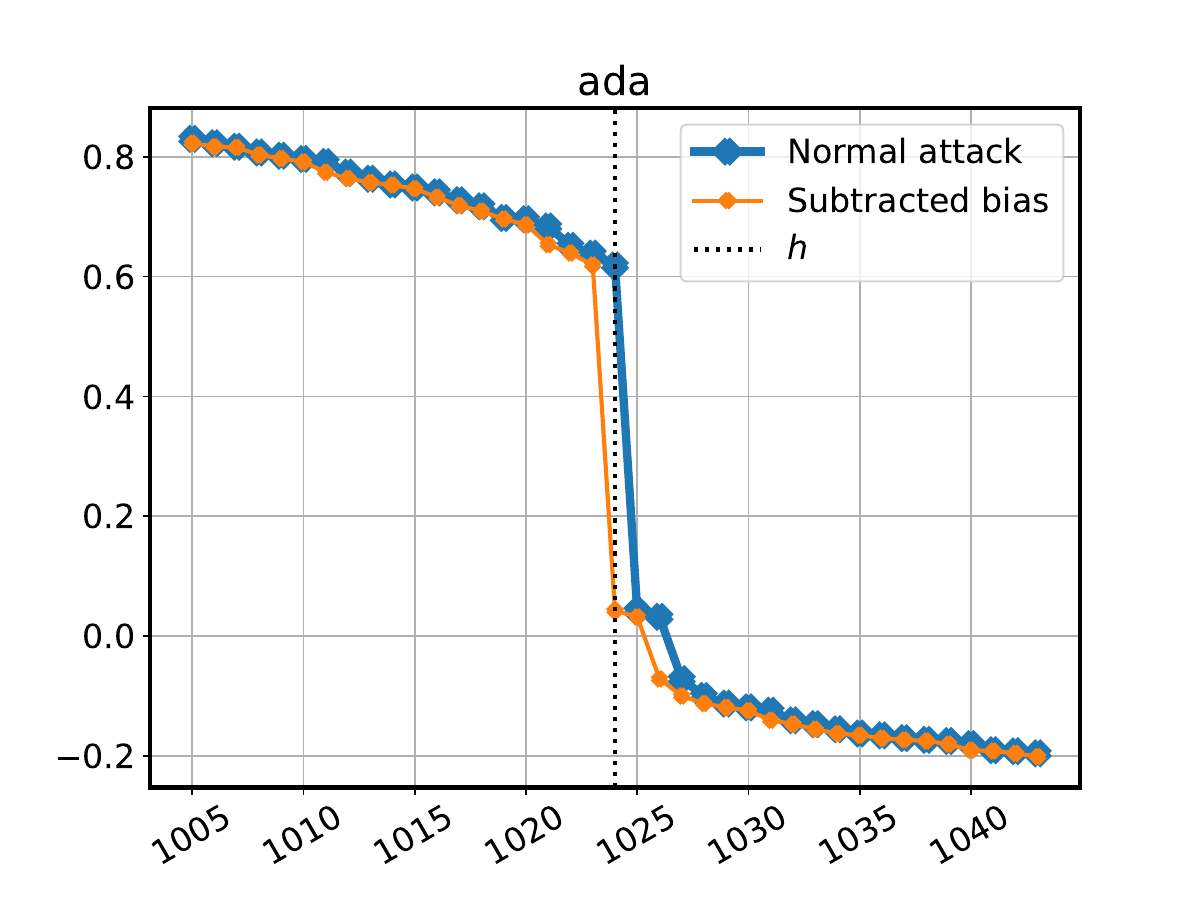}
    \footnotesize{(a) \texttt{ada} uses \textbf{LayerNorm}.}
  \end{minipage}
  \hfill
  \begin{minipage}{0.3\textwidth}
    \centering
    \includegraphics[width=1.15\linewidth]{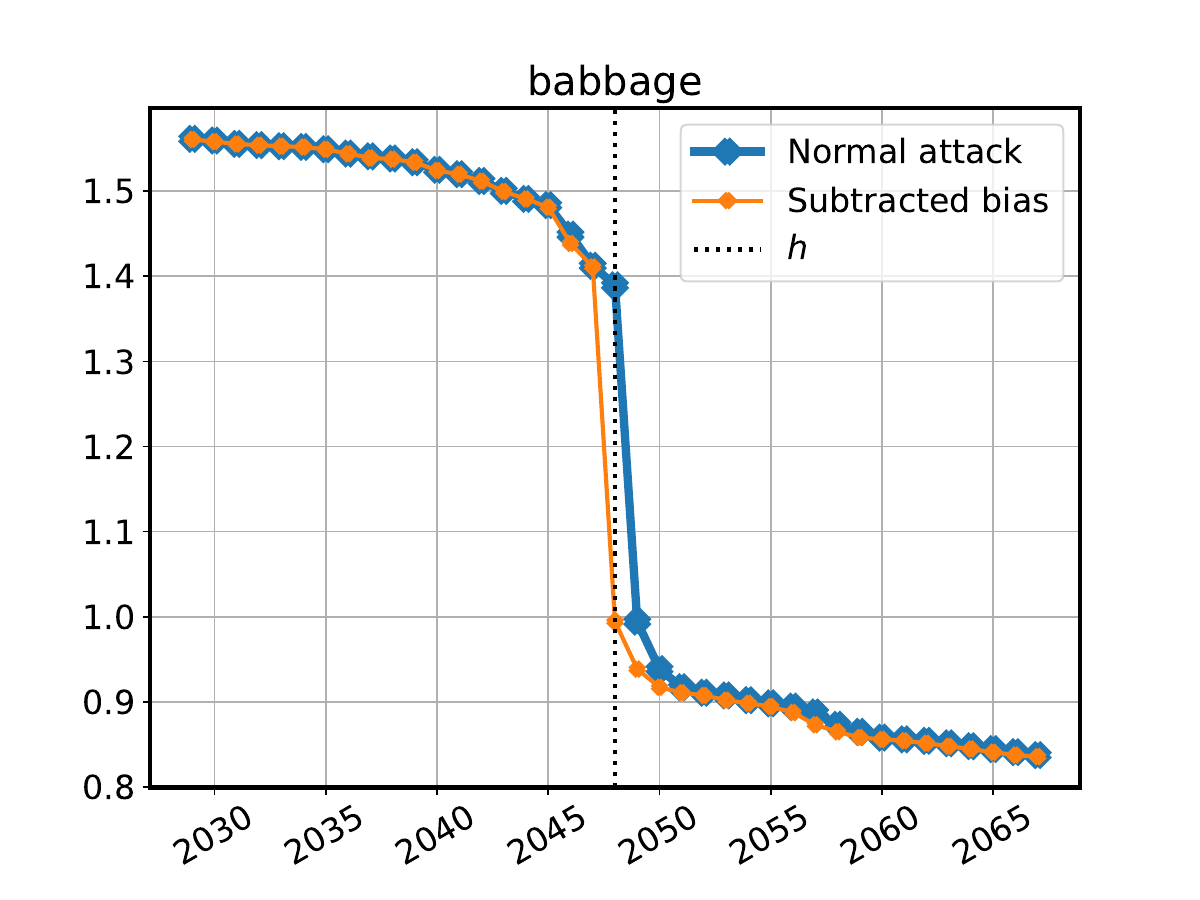}
    \footnotesize{(b) \texttt{babbage} uses \textbf{LayerNorm}.}
  \end{minipage}
  \hfill
  \begin{minipage}{0.3\textwidth}
    \centering
    \vspace{-.1cm}
    \includegraphics[width=1.15\linewidth]{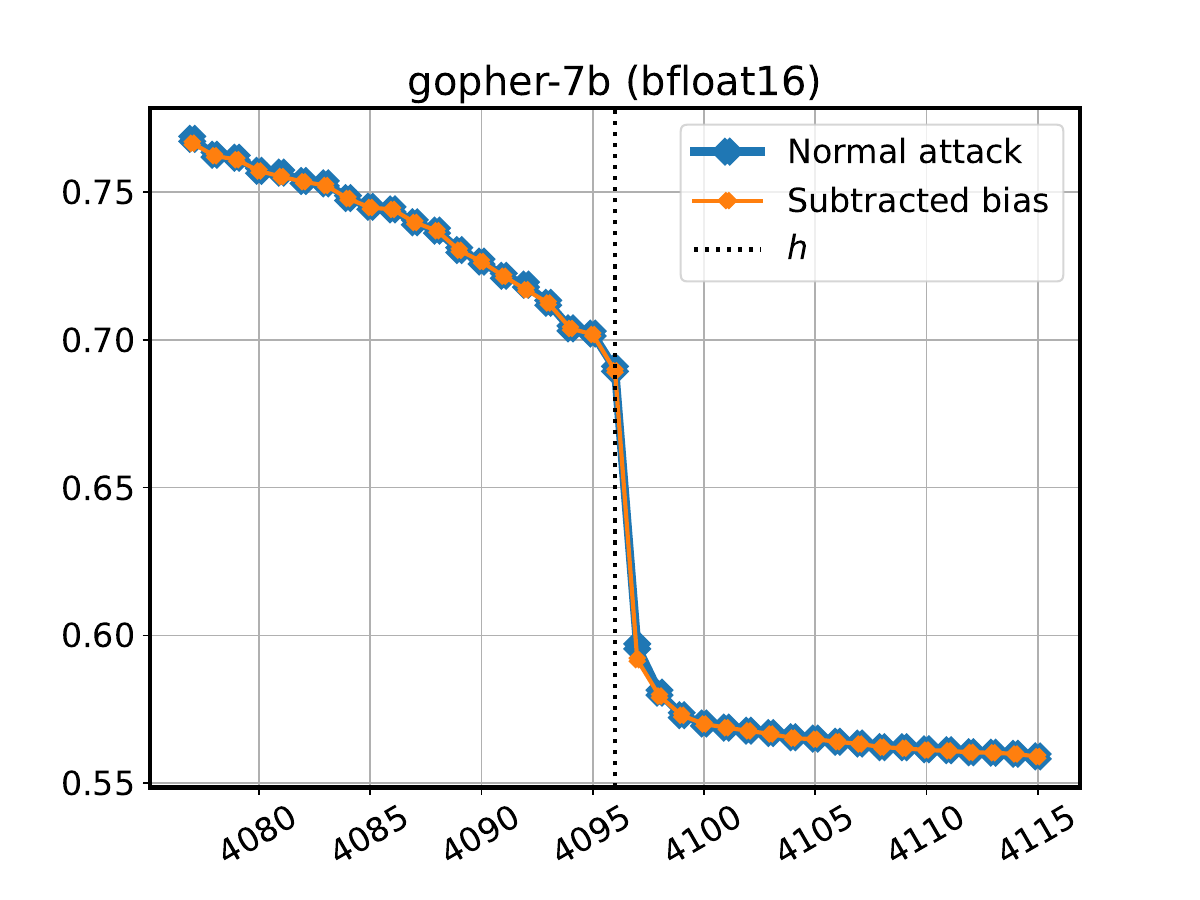}
    \footnotesize{(c) Gopher-7B uses \textbf{RMSNorm}.}
    \label{subfig:gopher_ln}
  \end{minipage}
  \caption{Stress-testing the LayerNorm extraction attack on models behind an API (a-b), and models using both RMSNorm and biases (c).}
  \label{fig:scaling-stealing-norm}
\end{figure}

As a final stress test, we found that all open language models that use RMSNorm do not use any bias terms \citep{stellabidermangooglesheet}. Therefore, we checked that our attack would not give a false positive when applied to a model with RMSNorm but with biases. We chose Gopher-7B \citep{gopher}, a model with public architectural details but no public weight access, that uses RMSNorm but also biases (e.g. on the output logits). In \Cref{fig:scaling-stealing-norm}c we show that indeed the $h$th singular value does not decrease for this model that uses RMSNorm.

%
%

\section{Proof of Lemma~\ref{lemma:recover-E-up-to-rotation}}
\label{sec:proof_of_42}

Restating the lemma from \Cref{sec:attack_detail}:

\custompar{\Cref{lemma:recover-E-up-to-rotation}} \textit{In the logit-API threat model, under the assumptions of Lemma ~\ref{lem:h_recovery}: \textnormal{(i)}
The method from \Cref{sec:attack_detail} recovers $\tEt = \Et \cdot \mathbf{G}$ for some $\mathbf{G} \in \R^{h \times h}$; \textnormal{(ii)}
With the additional assumption that $g_{\theta}(p)$ is a transformer with residual connections, it is impossible to extract $\Et$ exactly.}

We first give a short proof of (i): 

\begin{proof}
(i) To show we can recover $\tEt = \Et \cdot \mathbf{G}$, recall Lemma~\ref{lem:h_recovery}: we have access to $\Q^\top = \Et \cdot \H$ for some $\H \in \R^{h \times n}$.
Using the compact SVD of $\Q$ from the method in \Cref{sec:attack_detail}, $\Et \cdot \H \cdot \V = \U \cdot \mathbf{\Sigma}$.
We know $\mathbf{G} := \H \cdot \V \in \R^{h \times h}$,
hence if we take $\tEt = \U \cdot \mathbf{\Sigma}$, we have $\tEt = \Et \cdot \mathbf{G}$.
\end{proof}

Proving \Cref{lemma:recover-E-up-to-rotation}(ii) requires several steps due to the complexity of the transformer architecture: we progressively strengthen the proof to apply to models with no residual connections (\ref{subapp:e_with_fc}), models with residual connections (\ref{subapp:e_up_to_rotation_no_ln}), models with RMSNorm (\ref{subapp:specialize_orthog}), LayerNorm (\ref{subapp:extend_to_layer_norm}) and normalization with an $\varepsilon$ term (\ref{subapp:extend_to_eps_non_zer}).

\subsection{Proof of Lemma~\ref{lemma:recover-E-up-to-rotation}(ii) in Models With Fully-connected Layers}
\label{subapp:e_with_fc}

\begin{proof}[Proof of \Cref{lemma:recover-E-up-to-rotation}(ii)] As a gentle warmup, we prove (ii) under the additional assumption that the model does not use normalization layers (LayerNorm or RMSNorm) in its architecture. 
To prove (ii) we show it is possible to find a two distinct sets of model parameters $\theta, \theta'$ with different embedding projection matrices that result in identical API outputs.

We begin with a simpler case where $g_\theta$ does not have residual connections but a fully connected (FC) final layer. In this case, for any invertible $h \times h$ matrix $\bS$, we have that $g_\theta\br{p} = \bS g_{\theta^\prime}\br{p}$ where $\theta^\prime$ is the same as $\theta$ except that the weights of the final FC layer are pre-multiplied by ${\bS}^{-1}$. Hence, if $g_\theta$ has a final FC layer, it is impossible to distinguish between the embedding projection layer $\Et$ acting on $g_\theta$ and the embedding projection layer $\Et \cdot \bS$ acting on $g_{\theta'}$, given access to the output of the API $\api$ only.

\subsection{Proof of Lemma~\ref{lemma:recover-E-up-to-rotation}(ii) With Residual Layers}
\label{subapp:e_up_to_rotation_no_ln}

More generally, if $g_\theta$ is composed of residual layers but no normalization layers then $g_{\theta}(p) = \sum_i L_i (p)$, where $L_i(p)$ is the output of the $i$th residual layer in the model, ignoring the skip connection \citep{elhage2021mathematical, resnet_shallow_path}. Assume also that each $L_i$ has a final layer that is a fully connected linear layer and a linear input layer (this assumption is true for both attention and MLP modules in transformers without normalization layers). Constructing $\theta'$ such that each $L_i$ has input weights pre-multiplied by $\bS^{-1}$ and output FC weights multiplied by $\bS$, we have $g_{\theta'}(p) = \sum_i \bS L_i(p) = \bS \cdot g_{\theta}(p)$ by linearity. Finally, by using a new embedding projection matrix $(\bS^{-1})^\top \cdot \E$ and calculating
\begin{equation}\label{eqn:4_2_lemma_equation}
((\bS^{-1})^\top \cdot \E)^\top \cdot g_{\theta'}(p) = \Et \cdot g_{\theta}(p),
\end{equation}
we have shown that logit outputs are identical and so again we cannot distinguish these transformers by querying $\api$ and $\api'$ alone.



\subsection{Normalization Layers and Orthogonal Matrices}
\label{app:orthogonal_matrices}

In Sections \ref{app:orthogonal_matrices}-\ref{subapp:extend_to_eps_non_zer} we can no longer use general invertible matrices $\bS$ in our arguments, and must instead use orthogonal matrices, matrices $\U$ such that $\U^\top \U = I$. In models with LayerNorm, we specialise further, too (\Cref{subapp:extend_to_layer_norm}).



\begin{lemma}
The RMSNorm operation is equal to $x \mapsto \W n(x) + b$ where $\W$ is a diagonal matrix. 
\end{lemma}
\begin{proof}
RMSNorm is conventionally written as 
\begin{equation}
    x \mapsto \frac{w \cdot x}{\sqrt{\frac{1}{h} \sum_i x_i^2}} + b
\label{eqn:rms_norm}
\end{equation}
where $w$ is multiplied elementwise by normalized $x$. Clearly this can be written as a diagonal matrix. Further, we can multiply this diagonal matrix by $\sqrt{h}$ to cancel that factor in the denominator of \Cref{eqn:rms_norm}. Since $n(x) = x / ||x|| = x \bigg/ \sum_i \sqrt{x_i^2}$ we get the result.
\end{proof}

Intuitively, the proof in \Cref{subapp:e_up_to_rotation_no_ln} relied on pre-multiplying the input projection weight of layers by a matrix $\bS^{-1}$, so that this cancelled the rotation $\bS$ applied to the model's hidden state (called the `residual stream' in mechanistic interpretability literature \citep{elhage2021mathematical}). Formally, if we let the input projection layer be $\M$, we were using the fact that $\left( \M \bS^{-1} \right) \left(\bS x \right) = \M x$. However, since models with normalization layers use these before the linear input projection, the result of applying $\bS$ to the hidden state, if we apply the same procedure, produces the activation

\begin{equation}
    (\M \bS^{-1} ) (\W n(\bS x) + b)
    \label{eqn:norm_s}
\end{equation}
but since in general $n$ and $\bS$ do not commute, we cannot conclude that the $\bS$ transformations preserve the transformer's outputs. We will show that if we take $\bS = \U$ an orthogonal matrix, then we still get a general impossibility result.

To do this, we will need a simple result from linear algebra:

\begin{lemma}\label{lemma:linear-algebra-orthogonal}
    Let $x \in \mathbb{R}^h$. Then the normalization map $n(x) := \frac{x}{||x||}$ commutes with orthogonal matrices $\U$.
\end{lemma}

\begin{proof}[Proof of Lemma~\ref{lemma:linear-algebra-orthogonal}]
We need to show that $\frac{\U x}{||x||} = \frac{\U x}{||\U x||}$. This is true since $x^\top \U^\top \U x =  x^T x$, so $||\U x|| = ||x||$.
\end{proof}


\subsection{Proof of Lemma~\ref{lemma:recover-E-up-to-rotation}(ii) in Models With RMSNorm}
\label{subapp:specialize_orthog}

In \Cref{lemma:linear-algebra-orthogonal}, we showed that orthogonal matrices $\U$ commute with normalization. Hence if we multiply all layer output weights by $\U$, but pre-multiply all layer input projection weights by $\W \U^\top \W^{-1}$, then the effect of the linear projection layer is 

\begin{equation}
    (\M \W \U^\top \W^{-1}) (\W n(\U x) + b) = (\M \W \U^\top \W^{-1}) (\W\U n(x) + b) = \M ( \W n(x) + b )
\end{equation}

which is identical to the original model. Applying this procedure to all layers added to the hidden state (using the different $W$ diagonal matrices each time) gives us a model $g_{\theta'}(p)$ such that $g_{\theta'}(p) = \U g_{\theta'}(p)$ so a different embedding projection matrix $\Et \U^\top$ will give identical outputs to the original model $g_{\theta}(p)$ (with embedding projection $\Et$).

Note that we ignore what happens to $b$ in the above arguments, since any sequence of affine maps applied to a constant $b \in \R^h$ yields a constant $b' \in \R^h$, and we can just use $b'$ instead of $b$ in $g_{\theta'}$.

\subsection{Proof of Lemma~\ref{lemma:recover-E-up-to-rotation}(ii) in Models With LayerNorm}
\label{subapp:extend_to_layer_norm}

The LayerNorm operation is the composition of a centering operation $x \mapsto x - \bar{x}$ with RMSNorm (i.e. first centering is applied, then RMSNorm). Therefore the identical argument to \Cref{subapp:specialize_orthog} goes through, besides the fact that we need $\U$ to also commute with the centering operation. 
Since the centering operation fixes a $(h-1)$ dimensional subspace defined by $\mathbf{1}^T x = 0$ where $\mathbf{1} \in \R^h$ is the vector of ones,
it is enough to impose an additional condition that $\U \, \mathbf{1} \in \{-\mathbf{1}, \mathbf{1}\}$.

\subsection{Proof of Lemma~\ref{lemma:recover-E-up-to-rotation}(ii) in Models With Normalization $\varepsilon \neq 0$}
\label{subapp:extend_to_eps_non_zer}

We now extend to realistic models where the $\varepsilon$ in the denominator of LayerNorm is not 0. We can do this because the only fact we used about $x \mapsto n(x)$ was that $x \mapsto n(\U x)$ was identical to $x \mapsto \U n(x)$. In turn \Cref{lemma:linear-algebra-orthogonal} relied on $||\U x|| = ||x||$ due to orthogonality. But adjusting $n(x)$ to $n'(x) := x \bigg/ \sqrt{\frac1h ||x||^2 + \varepsilon}$ (i.e. normalization with an epsilon), since $||x|| = ||\U x||$, $n'$ commutes with $\U$, and so the proofs in \Cref{subapp:specialize_orthog} and \Cref{subapp:extend_to_layer_norm} still work when using $n'$ instead of $n$.

    %


Therefore finally, we have proven the impossibility result \Cref{lemma:recover-E-up-to-rotation}(ii) in all common model architectures (all non-residual networks that end with dense layers, and all transformers from \citet{stellabidermangooglesheet}).
\end{proof}

\section{Derivation of Binarized Logprob Extraction (Section~\ref{sec:binarized})}
\label{sec:binarized:proof}
To begin, observe that we can write
\begin{align*}
   y_{\text{top}} & = \logit_{\text{top}} - \log \sum_{i} \exp(\logit_i) \\
   y'_{\text{top}} & = \logit_{\text{top}} - \log \big( \exp(\logit_t-1) + \sum_{i \ne t} \exp(\logit_i) \big) 
\end{align*}

Let $\N = \sum_i \exp\br{\logit_i}$ and $p=\exp\br{\logit_t} / {\N}$.
Then, we can rewrite
\begin{align*}
    y_{\text{top}}  &= \logit_{\text{top}} - \log \N \\
    y_{\text{top}} &= \logit_{\text{top}} - \log(\N + (1/e-1) p \N)
\end{align*}
Subtracting the two, we get
\begin{align*}
    y_{\text{top}} - y'_{\text{top}} &= \log\br{1 + (1/e - 1) p} \\
\implies 
    p &= \frac{\exp(y_{\text{top}} - y'_{\text{top}}) - 1}{1/e - 1}.
\end{align*}

\paragraph{Related work.}
Concurrent work \cite{morris2023language} discusses a similar but weaker two-query logprob extraction.
Their attack requires a logit bias larger than $\logit_{\text{top}} - \logit_i$ and top-2 logprob access;
our attack works as soon the logit bias is allowed to be nonzero, and with top-1 logprob access.

\newcommand{\logits}{z}

\section{Efficient Recovery of Logits From Top $k$ Logprobs APIs} \label{sec:linear_reconstruction}

%


In Section \ref{sec:topk-extraction-main} of the main body, we presented a simple and practical method for extracting the entire logits vector via multiple queries to an API that only provides the top few logprobs and accepts a logit bias with each query.
In this section we present more efficient methods.

The method we presented earlier uses a reference token. We set this to some arbitrary value (e.g., $0$) and then compare the logits for all other tokens to this one. This approach is numerically stable, but is slightly wasteful: of the top $K$ logprobs returned by the API, one is always the reference token. Hence, we only recover $K-1$ logits per query with this method.

In this appendix, we present linear algebraic methods that are able to recover $K$ logits per query to the top-$K$ logprobs API.

\textbf{Setting:}
Recall that there is an unknown vector $\logits = \Et \cdot g_\theta(p) \in \mathbb{R}^\ell$ (i.e., the logits for a given prompt $p$) that we want to recover.
We can make multiple queries to the API with the same prompt $\api(p,b)$.
Each query is specified by a vector $\logitbias \in \mathbb{R}^\ell$ (a.k.a.~the logit bias).
We receive answers of the form $(i, a_i(\logits,\logitbias)) \in \mathbb{N} \times \mathbb{R}$, where $i$ is a token index and $a_i(\logits,\logitbias)$ is a logprob:
\begin{equation}
    a_i(\logits,\logitbias) = \log\left(\frac{\exp(\logits_i+\logitbias_i)}{\sum_j^\ell \exp(\logits_j+\logitbias_j)} \right) = \logits_i + \logitbias_i - \log\left(\sum_j^\ell \exp(\logits_j+\logitbias_j)\right).
    \label{eq:answertoquery}
\end{equation}
Each query may receive multiple answers (namely, the $K$ largest $a_i(\logits,\logitbias)$ values).
For notational simplicity, we denote multiple answers to one query the same way as multiple queries each returning one answer.
Suppose queries $\logitbias^1, \cdots, \logitbias^m$ were asked and we received $m$ answers $(i_1,a_{i_1}(\logits,\logitbias^1)) \gets \api(p,\logitbias^1), \cdots, (i_m,a_{i_m}(\logits,\logitbias^m)) \gets \api(p,\logitbias^m)$.

Our goal is to compute $\logits$ from the answers $a_i(\logits,\logitbias)$.

\subsection{Warmup: Single Logprob API ($K=1$)}
\label{sec:single-logprob}
As a starting point, suppose the API only returns the single largest logprob (i.e., $K=1$).
The approach from Section \ref{sec:topk-extraction-main} cannot work in this setting because we cannot obtain the logprob of both the reference token and another token at the same time, meaning we can recover less than $1$ logit per query.

The high-level idea to overcome this problem is that, instead of normalizing logits relative to a reference token, we shall normalize the logits to be logprobs. That is, we recover the logits with the normalization $\sum_j \exp(\logits_j) = 1$.
With this normalization it is no longer necessary to include a reference token in every query. 

Fix a token index $i$ and let $b_i=B$ and $b_j=0$ for all $j \ne i$. We query the API with this logit bias and assume that $B$ is large enough that token $i$ is returned: \[(i,a_i(\logits,\logitbias)) \gets \api(p,b).\]

From Equation \ref{eq:answertoquery},
\begin{align*}
    a_i(\logits,\logitbias) &= \logits_i + b_i - \log\left(\sum_j^\ell \exp(\logits_j+\logitbias_j)\right) \\
    &= \logits_i + B - \log\left(\exp(\logits_i+B) +\sum_{j \ne i} \exp(\logits_j)\right)\\
    &= \logits_i + B - \log\left(\exp(\logits_i+B)-\exp(\logits_i) + \sum_j^\ell \exp(\logits_j)\right),\\ \implies
    \logits_i + B - a_i(\logits,\logitbias) &= \log\left(\exp(\logits_i+B)-\exp(\logits_i) + \sum_j^\ell \exp(\logits_j)\right),\\ \implies
    \exp(\logits_i + B - a_i(\logits,\logitbias)) &= \exp(\logits_i+B)-\exp(\logits_i) + \sum_j^\ell \exp(\logits_j),\\ \implies
    \exp(\logits_i + B - a_i(\logits,\logitbias)) - \exp(\logits_i+B)+ \exp(\logits_i) &= \sum_j^\ell \exp(\logits_j),\\ \implies
    \exp(\logits_i) \cdot \left( \exp(B - a_i(\logits,\logitbias)) - \exp(B) + 1 \right) &= \sum_j^\ell \exp(\logits_j),\\ \implies
    \exp(\logits_i) &= \frac{\sum_j^\ell \exp(\logits_j)}{\exp(B - a_i(\logits,\logitbias)) - \exp(B) + 1},\\ \implies
    \logits_i &= \log\left(\sum_j^\ell \exp(\logits_j)\right) - \log\left(\exp(B - a_i(\logits,\logitbias)) - \exp(B) + 1\right).
\end{align*}
Thus if we normalize $\sum_j^\ell \exp(\logits_j)=1$, we have
\begin{equation}
    \logits_i = - \log\left(\exp(B - a_i(\logits,\logitbias)) - \exp(B) + 1\right). \label{eq:K1recovery}
\end{equation}

\subsection{Recovering $K$ Logits From $K$ Logprobs}

The approach from the previous subsection extends to the setting where each API query returns the top $K$ logprobs. In practice we work with $K=5$. We are able to recover $K$ logits. Again, instead of using a reference token to normalize the logits, we will normalize $\sum_j \exp(\logits_j) = 1$. However, in this setting we will need to solve a $K$-by-$K$ system of linear equations.

Fix $K$ token indices $i_1,\cdots,i_K$ and let $b_{i_k}=B$ for $k\in\{1,\cdots,K\}$ and $b_j=0$ for all $j \notin \{i_1,\cdots,i_K\}$. We query the API with this logit bias and assume that $B$ is large enough that the logprobs for $i_1,\cdots,i_K$ are returned as the top $K$ logprobs: \[(i_1,a_{i_1}(\logits,\logitbias)), (i_2,a_{i_2}(\logits,\logitbias)), \cdots ,(i_K,a_{i_K}(\logits,\logitbias))  \gets \api(p,b).\]
Let $\logits \in \mathbb{R}^\ell$ be the (unknown) logits and let $\N = \sum_{i} \exp(\logits_i)$ be the normalizing constant.
For each $k \in \{1,\cdots,K\}$, we have
\begin{align*}
    a_{i_k}(\logits,\logitbias) &= \logits_{i_k} + B - \log\left(\sum_{i \in \{i_1,\cdots,i_K\}} \exp(\logits_i + B) + \sum_{i \notin \{i_1,\cdots,i_K\}} \exp(\logits_i)\right) \\
    &=  \logits_{i_k} + B - \log\left((e^B-1)\sum_{i \in \{i_1,\cdots,i_K\}} \exp(\logits_i) + \sum_{i}^\ell \exp(\logits_i)\right) \\
    &=  \logits_{i_k} + B - \log\left((e^B-1)\sum_{i \in \{i_1,\cdots,i_K\}} \exp(\logits_i) + \N \right) , \\ \implies
      \logits_{i_k} + B - a_{i_k}(\logits,\logitbias) &= \log\left((e^B-1)\sum_{i \in \{i_1,\cdots,i_K\}} \exp(\logits_i) + \N \right) , \\ \implies
      \exp(\logits_{i_k} + B - a_{i_k}(\logits,\logitbias)) &= (e^B-1)\sum_{i \in \{i_1,\cdots,i_K\}} \exp(\logits_i) + \N , \\
\end{align*}
And therefore we can conclude
\[
      \exp(B - a_{i_k}(\logits,\logitbias)) \cdot \exp(\logits_k) - (e^B-1)\sum_{i \in \{i_1,\cdots,i_K\}} \exp(\logits_i) = \N .
      \]
This linear system of equations can be expressed in matrix form:
\begin{align*}
  A \cdot \left(\begin{array}{c} \exp(\logits_{i_1}) \\ \exp(\logits_{i_2}) \\ \vdots \\ \exp(\logits_{i_K}) \end{array}\right) = \left(\begin{array}{c} \N \\ \N \\ \vdots \\ \N \end{array}\right) , \\
\end{align*}
where $A$ is a $K \times K$ matrix with entries
\begin{align*}
  A_{k,j} = \begin{cases}
    \exp(B-a_{i_k}(\logits,\logitbias))-(e^B - 1) & \text{if $j = k$} \\
    -(e^B - 1) & \text{if $j \neq k$}.
  \end{cases}
\end{align*}
Note that $A$ is a rank-one perturbation of a diagonal matrix, that is, if $\one$ is the all-ones vector, then
\begin{align*}
  A =  \diag_{1 \le k \le K}(\exp(B-a_{i_k}(\logits,\logitbias))) - (e^B - 1)\one \one^T,
\end{align*}
where $\diag_{1 \le k \le K}(\exp(B-a_{i_k}(\logits,\logitbias)))$ denotes a diagonal matrix with the $k$-th diagonal entry being $\exp(B-a_{i_k}(\logits,\logitbias))$.
Inverting a diagonal matrix is easy and thus we can use the Sherman-Morrison formula to compute the inverse of $A$:
\begin{align*}
  A^{-1} 
  &= \diag_{1 \le k \le K}(\exp(a_{i_k}(\logits,\logitbias)-B))) \\ &+ (e^B-1)
  \frac{\diag_{1 \le k \le K}(\exp(a_{i_k}(\logits,\logitbias)-B))) \one \one^T \diag_{1 \le k \le 5}(\exp(a_{i_k}(\logitbias)-B)))}
       {1 - (e^B - 1) \one^T \diag_{1 \le k \le 5}(\exp(a_{i_k}(\logitbias)-B))) \one} \\
  &= \diag(v) + (e^B-1) \frac{ v v^T}{1 - (e^B - 1) \one^ T  v} ,
\end{align*}
where $v \in \mathbb{R}^K$ is the vector with entries $v_k = \exp(a_{i_k}(\logits,\logitbias) - B)$.
Hence
\begin{align*}
    \left(\begin{array}{c} \exp(\logits_{i_1}) \\ \exp(\logits_{i_2}) \\ \vdots \\ \exp(\logits_{i_K}) \end{array}\right) &=
  A^{-1} \cdot \left(\begin{array}{c} \N \\ \N \\ \vdots \\ \N \end{array}\right) \\ 
  &= \left( \diag(v) + (e^B-1) \frac{ v v^T}{1 - (e^B - 1) \one^ T  v} \right) \cdot \one \cdot \N \\
  &= \left( v + \frac{(e^B - 1) v v^T \one }{1 - (e^B - 1) \one^T  v} \right) \cdot \N \\
  &= \left( 1 + \frac{(e^B - 1) \one^T v }{1 - (e^B - 1) \one^T  v} \right) \cdot \N \cdot v\\
  &= \frac{\N}{1 - (e^B - 1) \sum_j v_j} \cdot v,\\
\implies
  \logits_{i_k} &= \log \left(A^{-1} \one \N\right)_k \\
  &= \log\left(\frac{\N v_k}{1 - (e^B - 1) \sum_j^K v_j} \right)\\
  &= \log\left(\frac{\N \exp(a_{i_k}(\logits,\logitbias) - B)}{1 - (e^B - 1) \sum_j^K \exp(a_{i_j}(\logits,\logitbias) - B)} \right)\\
  &= \log \N + a_{i_k}(\logits,\logitbias) - B - \log\left(1 - (e^B - 1) \sum_j^K \exp(a_{i_j}(\logits,\logitbias) - B)\right) \\
  &= \log \N + a_{i_k}(\logits,\logitbias) - B - \log\left(1 - (1 - e^{-B}) \sum_j^K \exp(a_{i_j}(\logits,\logitbias))\right) .
\end{align*}
If we normalize $\N=1$, this gives us a formula for computing the logits: 
\begin{equation}
    \logits_{i_k} = a_{i_k}(\logits,\logitbias) - B - \log\left(1 - (1 - e^{-B}) \sum_j^K \exp(a_{i_j}(\logits,\logitbias))\right) . \label{eq:Krecovery}
\end{equation} 
Note that setting $K=1$ yields the same result as in Equation \ref{eq:K1recovery}.

Recovery using Equation \ref{eq:Krecovery} is more efficient than the method in Section \ref{sec:topk-extraction-main}, as we recover $K$ logits $z_{i_1}, z_{i_2}, \cdots, z_{i_K}$ rather than just $K-1$ logits. However, if $B$ is large, numerical stability may be an issue. (And, if $B$ is small, the logit bias may be insufficient to force the API to output the desired tokens by placing them in the top $K$.)
Specifically, as $B \to \infty$, we have $(1-e^{-B}) \sum_j^K \exp(a_{i_j}(\logits,\logitbias)) \to 1$ and so the logarithm in Equation \ref{eq:Krecovery} tends to $\log(1-1)=-\infty$; this means we may have catastrophic cancellation.

\paragraph{Related work.}
Two works published during the responsible disclosure period use a similar procedure, and deal with numerical issues in different ways.
\cite{openlogprobs} 
start with a low $B$ for the whole vocabulary, then increase $B$ and ask for all tokens that haven't appeared before, and repeat until all tokens are covered.
\cite{hayase2024query} use the method in \cref{sec:single-logprob}, and set $B = -\hat{z}_i$, where $\hat{z}_i$ is an estimate of $z_i$ inherent to their application. It is possible variants of this method have been discussed before our or these works, but we are not aware of further references.

\subsection{General Method}

In general, we may not have have full control over which logprobs the API returns or which logit bias is provided to the API. Thus we generalize the linear algebraic approach above to reconstruct the logits from arbitrary logit biases and tokens. 

 Suppose queries $\logitbias^1, \cdots, \logitbias^m$ were asked and we received $m$ answers $(i_1,a_{i_1}(\logits,\logitbias^1)) \gets \api(p,\logitbias^1), \ldots, (i_m,a_{i_m}(\logits,\logitbias^m)) \gets \api(p,\logitbias^m)$. (If a query returns multiple answers, we can treat this the same as multiple queries each returning one answer.)

As before, rearranging Equation \ref{eq:answertoquery} gives the following equations.
\begin{align*}
   & \forall k\in[m] ~~~ \exp(a_{i_k}(\logits,\logitbias^k_{i_k})) = \frac{\exp(\logits_{i_k}+\logitbias_{i_k}^k)}{\sum_j^\ell \exp(\logits_j+\logitbias_j^k)}.\\
   & \forall k\in[m] ~~~ \sum_j^\ell \exp(\logits_j+\logitbias_j^k) = \exp(\logits_{i_k}+\logitbias_{i_k}^k-a_{i_k}(\logits,\logitbias^k)).\\
   & \forall k\in[m] ~~~ \sum_j^\ell \exp(\logits_j) \cdot \exp(\logitbias_j^k) = \exp(\logits_{i_k}) \cdot \exp(\logitbias_{i_k}^k-a_{i_k}(\logits,\logitbias^k)).\\
    & \forall k\in[m] ~~~ \sum_j^\ell \left( \exp(\logitbias_j^k) - {\mathbb{I}[j=i_k]}\cdot\exp(\logitbias_{i_k}^k-a_{i_k}(\logits,\logitbias^k)) \right) \cdot \exp(\logits_j) = 0. \\
    & A \cdot \left(\begin{array}{c} \exp(\logits_1) \\ \exp(\logits_2) \\ \vdots \\ \exp(\logits_\ell) \end{array}\right) = \left(\begin{array}{c} 0 \\ 0 \\ \vdots \\ 0 \end{array}\right) , \\
    & \text{where}~~~ \forall k\in[m] ~\forall j\in[\ell] ~~~ A_{k,j} = \exp(\logitbias_j^k) \cdot \left(1-\mathbb{I}[j=i_k]\cdot\exp(-a_{i_k}(\logits,\logitbias^k))\right).
\end{align*}
Here $\mathbb{I}[j=i_k]$ is $1$ if $j=i_k$ and $0$ otherwise.
If $A$ is invertible, then this linear system can be solved to recover the logits $\logits$. 
Unfortunately, $A$ is not invertible: Indeed, we know that the solution cannot be unique because shifting all the logits by the same amount yields the exact same answers $a_i(\logits,\logitbias)=a_i(\logits+\mathbf{1},\logitbias)$. That is, we expect a one-dimensional space of valid solutions to $A \cdot \exp(\logits) = \mathbf{0}$.
To deal with this we simply add the constraint that $\logits_1=0$ or, equivalently, $\exp(\logits_1)=1$.  This corresponds to the system 
\[
    \widehat{A} \cdot \exp(\logits) = \left(\begin{array}{c} 1~0~\cdots~0 \\ A \end{array}\right) \cdot \left(\begin{array}{c} \exp(\logits_1) \\ \exp(\logits_2) \\ \vdots \\ \exp(\logits_\ell) \end{array}\right) = \left(\begin{array}{c} 1 \\ 0 \\ \vdots \\ 0 \end{array}\right).
\]
(We could also normalize $\sum_i^\ell \exp(\logits_i) = 1$. This corresponds to the first row of $\widehat{A}$ being all $1$s instead of one $1$.)
This is solvable as long as the augmented matrix has a nonzero determinant
\begin{equation}
    \mathrm{det}\left(\widehat{A}\right)=\mathrm{det}\left(\begin{array}{c} 1~0~\cdots~0 \\ A \end{array}\right) = \mathrm{det}(A_{1:m,2:\ell}).
\end{equation}
Here $A_{1:m,2:d}$ denotes $A$ with the first column removed. Note that we are setting $m=\ell-1$. This is the minimum number of query-answer pairs that we need. If we have more (i.e., $m \ge \ell$), then the system is overdetermined. Having the system be overdetermined is a good thing; the extra answers can help us recover the logprobs with more precision. The least squares solution to the overdetermined system is given by
\begin{equation}
    \widehat{A}^T \widehat{A} \cdot \left(\begin{array}{c} \exp(\logits_1) \\ \exp(\logits_2) \\ \vdots \\ \exp(\logits_\ell) \end{array}\right) = \widehat{A}^T \left(\begin{array}{c} 1 \\ 0 \\ \vdots \\ 0 \end{array}\right). \label{eq:lstsq-recovery}
\end{equation}
This provides a general method for recovering the (normalized) logits from the logprobs API.

\paragraph{Related work.}
\cite{zanella2021grey} have an almost identical method, although they operate in the setting of a publicly known encoder and reconstructing the last layer.

\newpage
\section{Extraction From Logprob-free APIs}
\label{sec:logprobs_free_attacks}
\label{sec:logprob-free-appendix}

A more conservative API provider may remove access to the combination of logit bias and logprobs entirely. Indeed, after disclosing our attack to OpenAI, they removed the ability for logit bias to impact the top logprobs---thus preventing the attacks from the prior sections. 
To exploit situations such at this,
we further develop several logprob-free attacks that recover the complete logit vector by performing
binary search on the logit bias vector, albeit at increased cost.
\footnote{We release supplementary code that deals with testing these attacks without direct API queries at \url{https://github.com/dpaleka/stealing-part-lm-supplementary}.}

\paragraph{API:} Some APIs provide access to a logit bias term, 
but do not provide any information about the logprobs. Thus, we have, 
\[\api(\x, \logitbias) = \Argtopk\br{\logsoftmax\br{\Et \cdot g_\theta(\x) + \logitbias}}.\]
where $\Argtopk\br{\z}$ returns the index of the highest coordinate in the vector $\z \in \R^l$. In this section, we will use the notation $b=\{i: z\}$ to denote that the bias is set to $z$ for token $i$ and $0$ for every other token. We also use $b=\{\}$ to denote that no logit bias is used. Finally, we assume that the bias is restricted to fall within the range $[-B, B]$.

\paragraph{What can be extracted?} The attacks developed in this Section reconstruct the logit vector up to an additive ($\infty$-norm) error of $\varepsilon$.

\subsection{Warm-up: Basic Logprob-free Attack}
\label{sec:logprobfree:basic}
\paragraph{Method.}
We make one simple insight for our logprob-free attacks: sampling with temperature 0 produces the token with the largest logit value.
By adjusting the logit bias for each token accordingly, we can therefore recover every token's logit value through binary search.
Formally, let $p$ be the prompt, and relabel tokens so that the token with index $0$ is the most likely token in the response to $p$, 
given by $\api(p, \logitbias=\{\})$.
For each token $i \neq 0$, we run a binary search over the logit bias term to find the minimal value $x_i \ge 0$ such that the model emits token $i$ with probability 1.
This recovers all logits (like all prior attacks, we lose one free variable due to the softmax).

\begin{algorithm}
\caption{Learning logit differences \label{alg:binsearch_logit}}
\begin{algorithmic}
\State $\alpha_i \gets -B, \beta_i \gets 0$
\While{$\beta_i - \alpha_i > \varepsilon$}
\If{$\api\br{p, \logitbias=\{i: -\frac{\alpha_i + \beta_i}{2}\}}=0$}
\State $\beta_i \gets \frac{\alpha_i + \beta_i}{2}$
\Else
\State $\alpha_i \gets \frac{\alpha_i + \beta_i}{2}$
\EndIf
\State Return $\frac{\alpha_i + \beta_i}{2}$
\EndWhile
\end{algorithmic}

\end{algorithm}

\paragraph{Analysis.} This attack, while
inefficient, correctly extracts the logit vector.

\begin{lemma}
For every token $i$ such that $\logit_i - \logit_0 \geq -B$, Algorithm \ref{alg:binsearch_logit} outputs a value that is at most $\varepsilon$ away from the $\logit_i - \logit_0$ in at most $\log\br{\frac{B}{\varepsilon}}$ API queries.
\end{lemma} 
\begin{proof}
 The API returns the (re-ordered) token $0$ as long as the logit bias added is smaller than $\logit_i - \logit_0$. By the assumption, we know that $\logit_i-\logit_0 \in [-B, 0]$. The algorithm ensures that $\beta_i \geq \logit_i - \logit_0 \geq \alpha_i$ at each iteration, as can be seen easily by an inductive argument. Further, $\beta_i-\alpha_i$ decreases by a factor of $2$ in each iteration, and hence at termination, we can see that the true value of $\logit_i-\logit_0$ is sandwiched in an interval of length $\varepsilon$. Furthermore, it is clear that the number of iterations is at most $\log_2\br{\frac{B}{\varepsilon}}$ and hence so is the query cost of this algorithm.
\end{proof}

\paragraph{Limitations of the approach.} If $\logit_i - \logit_0 < -2B$ it is easy to see there is no efficient way to sample the token $i$, hence no way to find information about $\logit_i$ without logprob access.
There is a way to slightly increase the range for $-2B \le \logit_i-\logit_0 \le -B$ by adding negative logit biases to the tokens with the largest logit values,
but we skip the details since for most models, for the prompts we use, the every token satisfies $\logit_i - \logit_0 > -B$.

\paragraph{Related work.}
Concurrent work \cite{morris2023language} has discussed this method of extracting logits.

\subsection{Improved Logprob-free Attack: Hyperrectangle Relaxation Center}
\label{sec:logprobfree:hyperrectangle-center}

We can improve the previous attack by modifying the logit bias of multiple tokens at once.
%

\paragraph{API:} We use the same API as in the previous section, with the additional constraint that the $\api$ accepts at most $N+1$ tokens in the logit bias dictionary. We again first run a query $\api(p, \logitbias=\{\})$ to identify the most likely token and set its index to $0$.
Our goal is to approximate $\logit_i - \logit_0$ for $N$ different tokens.
If $N < l - 1$, we simply repeat the same algorithm for different batches of $N$ tokens $\frac{l-1}{N}$ times.

\paragraph{Method.}
\begin{algorithm}
\caption{Learning logit differences with multi-token calls}\label{alg:binsearch_logit_multi}
\begin{algorithmic}
\State $\alpha_i \gets -B, \beta_i \gets 0 \quad \forall i=1, \ldots, N$
\State $\calC=\{\logit: \logit_i - \logit_0 \leq B \quad \forall i=1, \ldots, N \}$
\For{$T$ rounds}
\State $b_i \gets -\frac{\alpha_i + \beta_i}{2}$ for $i=0,\ldots, N$
\State $k \gets \api(p, \logitbias=\{{0: b_0, 1: b_1, \dots, N: b_N}\})$
\For{$j \neq k$}
\State $\calC \gets \calC \cap \{\logit: \logit_k + b_k
\geq \logit_j + b_j\}$
\EndFor
\For{$i=0,\ldots, N$}
\State $\alpha_i \gets \displaystyle\min_{\logit\in\calC} \logit_i-\logit_0$
\State $\beta_i \gets \displaystyle\max_{\logit \in \calC} \logit_i-\logit_0$
\EndFor
\EndFor
\State Return $[\alpha_i, \beta_i] \quad \forall i \in \{0, \ldots, N\}$
\end{algorithmic}

\end{algorithm}

Our approach queries the API with the logit bias set for several tokens in parallel. The algorithm proceeds in \emph{rounds}, where each round involves querying the API with the logit bias set for several tokens. 

Suppose that the query returns token $k$ as output when the logit bias was set to $\{i: b_i\}$ for $i=1,\ldots, l$ and the prompt is $p$. Then, we know that $\logit_k + b_k \ge \logit_j + b_j$ for all $j \neq k$ by the definition of the API. 

This imposes a system of linear constraints on the logits. By querying the model many times, and accumulating many such systems of equations, we can recover the logit values more efficiently. To do this, we accumulate all such linear constraints in the set $\calC$, and at the end of each round, compute the smallest and largest possible values for $\logit_i-\logit_0$ by solving a linear program that maximizes/minimizes this value over the constraint set $\calC$. Thus, at each round, we can maintain an interval that encloses $\logit_i-\logit_0$, and refine the interval at each round given additional information from that round's query. After $T$ rounds (where $T$ is chosen based on the total query budget for the attack), we return the tightest known bounds on each logit. 

\begin{lemma}
\label{lemma:query-linear-constraints}
Suppose that $\logit_i-\logit_0 \in [-B, 0]$ for all $i=1, \ldots, l$. Then, Algorithm \ref{alg:binsearch_logit_multi} returns an interval $[\alpha_i, \beta_i]$ such that $\logit_i-\logit_0 \in [\alpha_i, \beta_i]$ for each $i$ such that $\logit_i - \logit_0 \in [-B, 0]$. Furthermore, each round in the algorithm can be implemented in computation time $O(N^3)$ (excluding the computation required for the API call). 
\end{lemma}
\begin{proof}
Algorithm \ref{alg:binsearch_logit_multi} maintains the invariant that $\logit_i-\logit_0 \in [\alpha_i, \beta_i]$ in each round. We will prove by induction that this is true and that the true vector of logits always lies in $\calC$. Note that by the assumption stated in the Lemma, this is clearly true at the beginning of the first round. Suppose that this is true after $K < T$ rounds. Then, in the $K+1$-th round, the constraints added are all valid constraints for the true logit vector, since the API returning token $k$ guarantees that $\logit_k+\logitbias_k \ge \logit_j +\logitbias_j$ for all $j \neq k$. Hence, by induction, the algorithm always ensures that $\logit_i - \logit_0 \in [\alpha_i, \beta_i]$. 

In \Cref{sec:shortest-path},
we show the LP to compute $\alpha_i, \beta_i$ for all $i$ can be seen as an all-pairs shortest paths problem on graph with edge weights $c_{jk}=\min_{\text{rounds}} \logitbias_j-\logitbias_k$ where the minimum is taken over all rounds where the token returned was $k$. This ensures the computation complexity of maintaining the logit difference intervals is $O(N^3)$.
\end{proof}

\subsubsection{Shortest-path Formulation of the Logprob-free Attack LP}
\label{sec:shortest-path}

It is actually possible to improve the computational efficiency of the hyperrectangle relaxation of the polytope $\calC$.
Here we show how to formulate this problem as a shortest path problem on a weighted graph.
This enables us to quickly compute the exact $[\alpha_i, \beta_i]$ for all $i \in \{1, \dots, N\}$ after each query.
\begin{lemma}
\label{lemma:bounds-update}
Let $G = (\{0, 1, \dots, N\}, E)$ be a weighted directed graph without negative cycles.  Let $\mathcal{P} \subset \mathbb{R}^{n+1}$ be the solution set of a system of linear inequalities:
\begin{align*}
  \logit_i - \logit_j &\le c_{ji} 
  \quad \forall~ j \xrightarrow{c_{ji}} i  \quad\in E
\end{align*}
Then if $\logit_0 = 0$, we have
\begin{align*}
  \max_{x \in \calC} \logit_i &= \text{distance in } G \text{ from } 0 \text{ to } i.
\end{align*}
\end{lemma}
\begin{proof}
  Let $e_{0j_1}, e_{j_1j_2}, \dots, e_{j_{m-1}i}$ be the edges of the minimum distance path from $0$ to $i$ in $G$.
  We have 
  \begin{align*}
    \logit_i & \le \logit_{j_{m-1}} + c_{j_{m-1}i} \le \ldots \\
    & \le \logit_0 + \sum_{t=1}^{m-1} c_{j_{t+1}j_t} = \sum_{t=1}^{m-1} c_{j_{t+1}j_t},
  \end{align*}
  hence the shortest path is an upper bound on $\logit_i$. 
  To prove feasibility, we claim that setting $\logit_i$ to be the distance from $0$ to $i$ satisfies all the inequalities.
  Assume some inequality $\logit_i - \logit_j \le c_{ji}$ is violated. 
  Then we can go from $0 \to j \to i$ in $G$ with a total weight of $\logit_j + c_{ji} < \logit_i$, which contradicts the assumption that $\logit_i$ is the distance from $0$ to $i$.
\end{proof}

To apply this to our setting, note that (1) all constraints, even the initial $\alpha_i \le \logit_i \le \beta_i$, are of the required form;
(2) the graph has no negative cycles because the true logits give a feasible solution.
(3) we can get the lower bounds by applying the same procedure to the graph induced by inequalities on $-\logit_i$.

We can find the distances from $0$ to all other vertices 
using the Bellman-Ford algorithm in $O(N^3)$ time. If $N = 300$, this is at most comparable to the latency of $\api$.
Since only $N$ edges of the graph update at each step,
we note that the heuristic of just updating and doing a few incremental iterations of Bellman-Ford 
gets $[\alpha_i, \beta_i]$ to high precision in practice.
The number of API queries and the token cost, of course, remains the same.


\subsection{Improved Logprob-free Attack: Better Queries on Hyperrectangles}
\label{sec:logprobfree:hyperrectangle-fancy}

The main issue of the previous approach is that some tokens are sampled more often than others,
even in the case our prior for the $\logit$ vector is uniform over $[-B, 0]$.
This is because the "centering of the hyperrectangle" logit bias does not partition the hyperrectangle into equally-sized parts labeled by the argmax coordinate. 
For example, if  $\beta_i - \alpha_i \ll \beta_j - \alpha_j$, under an uniform prior over $[\alpha_i, \beta_i] \times [\alpha_j, \beta_j]$, $j$ will be much more likely to be the output token than $i$.
Hence, in Algorithm \ref{alg:binsearch_logit_multi} we rarely get constraints lower-bounding $\logit_i$ in terms of other logits, which makes for weaker relaxations of $\mathcal{C}$.

Our solution is to bias tokens so that the output token distribution is closer to uniform; in particular, biasing the token with the smallest $\beta_t - \alpha_t$ (the $0$ token) to have probability exactly $1/(N+1)$ given an uniform prior over the hyperrectangle.
One logit bias that satisfies this is:
\begin{align}
b_i &= -(1-c) \alpha_i - c \beta_i \quad \forall i=0,\ldots,N \notag \\ 
\text{ where } \quad c &= \exp(-\log(N+1) / N). \label{eq:one-over-n}
\end{align}
We now run Algorithm \ref{alg:binsearch_logit_multi}, with one simple modification: we replace $b_i = -\frac{\alpha + \beta}{2}$ with $b = - (1-c) \alpha - c \beta$.
As can be seen in Table \ref{tab:model_comparison_logit_estimation}, the modified algorithm outperforms the method in \ref{sec:logprobfree:hyperrectangle-center} significantly.

The goal of balanced sampling of all output tokens can be approached in many ways. For example, we could tune $c$ in the above expression;
bias tokens which $\api$ hasn't returned previously to be more likely; or
solve for the exact logit bias that separates $\calC$ (or some relaxation) into equal parts.
However, we show in \Cref{sec:howfaroptimal} that, under some simplifying assumptions, the queries/logit metric of this method in Table \ref{tab:model_comparison_logit_estimation} is surprisingly close to optimal.




\section{How Far Are Our Logprob-Free Attacks From Optimal?}
\label{sec:howfaroptimal}


In the logprob-free API, we have produced attacks capable of recovering logits and ultimately the embedding hidden dimension and embedding matrix up to a similarity transform. We now provide lower bounds on the minimum number of queries required by \emph{any} attacker attempting model stealing under the logprob-free API threat model.

\begin{lemma}
\label{lemma:information-lower-bound-logprobfree}
Assume the entries of $\logit \in \R^l$ are i.i.d. uniform over $[-B, 0]$. To recover the vector $\logit$ up to $\infty$-norm error $\varepsilon$, the number of queries to $\api(p, \cdot)$ we need is at least:
\[
\frac{\vocabsize \log_2(B/\varepsilon)}{\log_2(\vocabsize)}.
\]
\end{lemma}
\begin{proof}

The information content of a single logit value in $[-B, 0]$
up to $\infty$-norm error $\varepsilon$ is $\log_2(B/\varepsilon)$,
assuming a uniform prior over $\varepsilon$-spaced points in the interval.
Since the logits are independent, the information encoded in $\vocabsize$ logit values up to $\infty$-norm error $\varepsilon$ is $\vocabsize \log_2(100/\varepsilon)$.

Any single query to $\api$, no matter how well-crafted, yields at most $\log_2(\vocabsize)$ bits, because the output is one of $\vocabsize$ distinct values. The minimum number of queries required is at least the total information content divided by the information per query, yielding the lower bound
$
\vocabsize \log_2(B/\varepsilon) / \log_2(\vocabsize).
$
\end{proof}

The restriction of biasing at most $N$ tokens at a time gives us a lower bound of
\[
\frac{\vocabsize \log_2(B/\varepsilon)}{\log_2(N)}
\]
queries, which is a factor of $\log_2(\vocabsize) / \log_2(N)$ worse. 
For $N = 300$ and $\vocabsize \approx 100{,}000$, this is only a factor of $2$.

For $B = 100$ and $N = 300$, we thus need at least 
\[
\frac{\log_2(B / \varepsilon)}{\log_2(N)} \approx 0.81 + 0.12 \log_2(1/\varepsilon) 
\]
queries per logit. 
If we want between 6 and 23 bits of precision, the lower bound corresponds to 1.53 to 3.57 queries per logit.
We see that the best logprob-free attack in Table \ref{tab:model_comparison_logit_estimation} is only about 1 query per logit worse than the lower bound.

The main unrealistic assumption in Lemma \ref{lemma:information-lower-bound-logprobfree} is that the prior over the logit values is i.i.d. uniform over an interval. A better assumption might be that most of the logit values come from a light-tailed unimodal distribution. We leave more realistic lower bounds and attacks that make use of this better prior to future work.

\section{Recovering $\mathbf{W}$ up to an orthogonal matrix}
\label{app:recovery-orthogonal}

In this section, we present an algorithm for extracting $\Et$ up to an orthogonal $h \times h$ matrix, instead of a non-singular $h \times h$ matrix as in Appendix \ref{sec:proof_of_42}.
This algorithm requires solving a system of $O(h^2)$ linear equations and is hence prohibitive for large models in production, some of which have $h>1000$. However, we present proof that this technique works in practice in a notebook\footnote{See here: \href{https://colab.research.google.com/drive/1DY-q6rnokD6q8RQBwxVmAbtZ_85zq4_W?usp=sharing}{Recovering $W$ up to an orthogonal matrix (Colab notebook)}} on Pythia-14M, despite the  assumptions bellow (\Cref{subapp:assumptions}).

\subsection{Assumptions}
\label{subapp:assumptions}

We make a few simplifying assumptions:

\begin{enumerate}
    \item We merge the final normalization layer weights $\gamma$ into $\Et$ by linearity.\footnote{For a full explanation of this method of rewriting the unembedding matrix, see Appendix A.1, `Folding LayerNorm' in \citet{wesuniversal}. The intution is that $\gamma$ is another linear transformation applied just before $\W$ (a diagonal matrix). Therefore together they are together one matrix multiplication.}
    \item We assume that the output matrix always outputs 0-centered logits (as in \Cref{subsubsec:stealing-theory}). This isn't a restrictive assumption since model logprobs are invariant to a bias added to the logit outputs.
    \item We assume the numerical precision is high enough that during the final normalization layer, the hidden states are on a sphere.
    \item There is no degenerate lower-dimensional subspace containing all $g_{\theta}(p)$ for all our queries $p$.
    \item We assume the $\varepsilon$ in RMSNorm/LayerNorm is negligible.
\end{enumerate}

\renewcommand{\A}{\mathbf{A}}
\renewcommand{\O}{\mathbf{O}}
\newcommand{\X}{\mathbf{X}}
\renewcommand{\M}{\mathbf{M}}


\subsection{Methodology}

\paragraph{Intuition.} The main idea is to exploit the structure imposed by the last two layers of the model: LayerNorm/RMSNorm, which projects model internal activations to a sphere, and the unembedding layer. The model's logits lie on an ellipsoid of rank $h$ due to \Cref{imageellips}:

\begin{lemma}\label{imageellips}
The image of an ellipsoid under a linear transformation (such as the unembedding layer) is itself an ellipsoid.
\end{lemma}

Thus, we can project the model's logits outputs to that subspace using $\U$ the truncated SVD matrix; $\X=\U^\top \Q$.  An ellipsoid is defined with the quadratic form as in \Cref{eqnellips}: 
\begin{lemma}\label{eqnellips}
An ellipsoid is defined by a single semipositive-definite symmetric matrix $\A$ and a center by the following equation:
$(x-c)^\top \A (x-c)=1$.
\end{lemma}

By expanding this formula $ x^\top \A x -2\A c x + c^\top \A c = 1$, substituting $d = \A c$ we end up with a system that all (projected) model outputs should satisfy. This system is linear in $\binom{h+1}{2}$ entries of $\A$ (as $\A$ is symmetric), and the $h+1$ entries of $d$. The system doesn't constrain the free term $c^\top \A c$, and we can get rid of it by subtracting $x_0$ from other outputs of the model, effectively forcing the ellipsoid to pass through the origin, and thus satisfy $c^\top \A c = 1$. Using Cholesky decomposition on the PSD symmetric matrix $\A = \M \M ^\top$, we would end up with $\M$ a linear mapping from the unit-sphere to the (origin-passing) ellipsoid, which is equivalent to $\W$ up to an orthogonal projection (see \Cref{orthogw}). This should be centered at  $\W b = c - x_0 $.

\paragraph{Procedure.} The procedure we took in order to recover $ \W $ up to an orthogonal matrix as implemented in the notebook \ref{??} is described as follows at a high level (we expand in detail on several of the steps below):

\begin{enumerate}
    \item Collect $\binom{h+1}{2} + (h+1)$ logits and store them in a matrix $\mathbf{Q}$.
    \item Shift to origin: $\Q \gets \Q - \Q[0]$.
    \item Perform SVD on Q: $\U\mathbf{\Sigma}\V^\top=\Q$.
    \item Find the model's hidden dimension $h$ (e.g. using \Cref{sec:attack_detail}) and truncate $\U$
    \item Solve $ x^\top \A x -2dx = 0$ using SVD/QR to find the nullspace, such that
        \begin{itemize}
            \item $d = \A c$
            \item $c^\top \A c = 1$
            \item $\A$ is symmetric.
        \end{itemize}
    This process has time complexity: $O(h^6)$ (in many cases much faster in practice).
    \item Find $c=\A^{-1}d$, and scale $c$ and $\A$ to satisfy $ c^\top \A c = 1 $.
    \item Use Cholesky decomposition to find $\M$ s.t. $\M^\top\M=\A$.
    \item Obtain $\W = \U \cdot \M^{-1} \cdot \O$ for some orthogonal matrix $\O$.
\end{enumerate}



In steps 3-4, we use the compact SVD on the query output matrix $\Q = \U \cdot \mathbf{\Sigma} \cdot \Vt$.
Here $\Q \in \mathbb{R}^{l \times n}$, $\U \in \mathbb{R}^{l \times h}$, $\mathbf{\Sigma} \in \mathbb{R}^{h \times h}$, and $\Vt \in \mathbb{R}^{h \times n}$.
Note that the points $g_{\theta}(p)$ lie on a sphere in $\R^h$, and $\U^\top \cdot \Et \in \R^{h \times h}$,
hence $\U^\top \cdot \Et \cdot g_{\theta}(p)$ lie on an ellipsoid in $\R^h$. 
From now on, it is convenient to work with the points $\X = \U^\top \cdot \Q$;
As centered-ellipsoids are equivalently defined by $x^\top \A x = 1$ for some positive semidefinite (symmetric) matrix $\A \in \R^{h \times h}$,
this implies that we can write $\A = \M^\top \cdot \M$ for some $\M$ (step 7 of the procedure) which will work since $\A$ is positive semidefinite and symmetric since the system is non-degenerate because the ellipse is rank.
Importantly, to motivate step 8 of the procedure, we use \Cref{orthogw}.
\begin{lemma}\label{orthogw}
$\W = \U \cdot \M^{-1} \cdot \O$ for some orthogonal matrix $\O$.
\end{lemma}
\begin{proof}
We know that $g_\theta(p_i)$ lie on a sphere. The equation ($x_i-c)^\top \A (x_i-c) = 1$ is equivalent to ($x_i-c)^\top \M^\top \M (x_i-c) = 1$, which is equivalent to $\lVert \M (x_i-c) \rVert = 1$. 
This means that $\M (x_i-c)$ lie on a sphere. 
Because $\M (x_i-c) = \M \cdot \U^\top \cdot \Et \cdot g_{\theta}(p_i)$, we have that $\M \cdot \U^\top \cdot \Et$ is a norm-preserving transformation on the points $g_{\theta}(p_i)$.
By the assumption that $g_{\theta}(p_i)$ are not in a degenerate lower-dimensional subspace, we have that $\M \cdot \U^\top \cdot \Et =: \O$ is a norm-preserving endomorphism of $\R^h$,
hence an orthogonal matrix. This directly implies $\W = \U \cdot \M^{-1} \cdot \O$ as claimed.
\end{proof}




\section{Quantization and Noise}
\label{app:quantnoise}

\subsection{Quantization}
Quantization is a popular strategy for decreasing a model's memory footprint and speeding up inference. 
In addition to these benefits, using lower-precision number representations also effectively adds noise.
As noted in Section~\ref{sec:mitigations}, adding noise to the output logits could prevent our attack. A natural question that follows is, does quantization add sufficient noise to make our attack ineffective or more difficult to carry out? 

For a simple test, we quantize Llama-7B at both 8-bits and 4-bits, and compare our baseline attack (Section~\ref{sec:attack_warmup}) to the default 16-bit implementation.  
We quantize using \nicett{bitsandbytes}~\citep{dettmers2022optimizers}, which HuggingFace supports for out-of-the-box quantization of model weights and lower-precision inference (Figure~\ref{fig:llama-7b-quantized-comparison}).  
We observe no meaningful differences at different levels of quantization; querying each model results in recovering the same same embedding matrix dimension $h$ in the same number of queries.
Given that 8-bit and 4-bit quantization are generally observed to not have a large impact on performance, this is perhaps an unsurprising result; any noise from quanitization does not seem to have a meaningful impact on the logits (in the context of our attack).

\subsection{Noise}
One natural defense to our attacks is to obfuscate the logits by adding noise. This will naturally induce a tradeoff between utility and vulnerability---more noise will result in less useful outputs, but increase extraction difficulty. We empirically measure this tradeoff in Figure~\ref{subfig:noise-mse}. We consider noise added directly to the logits, that is consistent between different queries of the same prompt. To simulate this, we directly add noise to our recovered logits, and recompute the extracted embedding matrix. For GPT-2, we measure the RMSE between the true embedding matrix and the embedding matrix extracted with a specific noise level; for \ada{} and \babbage{}, we measure the RMSE between the noisy extracted weights and the weights we extracted in the absence of noise. We normalize all embedding matrices (to have $\ell_2$ norm 1) before measuring RMSE.

\begin{figure}
  \begin{minipage}{0.3\textwidth}
    \centering
    \includegraphics[width=\linewidth]{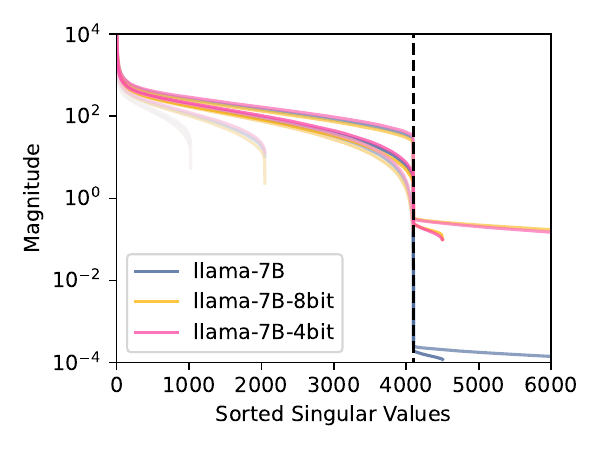}
    \captionof{subfigure}{Sorted singular values for $\{1024, 2048, 4096, 8192\}$  queries. 
    \label{subfig:llama-7b-quantized-queries}}
  \end{minipage}
  \hfill
  \begin{minipage}{0.3\textwidth}
    \centering
    \vspace{-.1cm}
    \includegraphics[width=\linewidth]{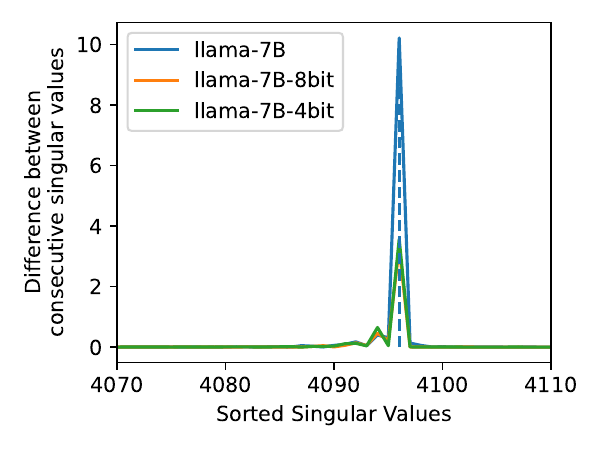}
    \captionof{subfigure}{Differences between consecutive sorted singular values.
    \label{subfig:llama-7b-quantized-s-diff}}
  \end{minipage}
  \hfill
  \begin{minipage}{0.3\textwidth}
    \centering
    \vspace{-.4cm}
    \includegraphics[width=\linewidth]{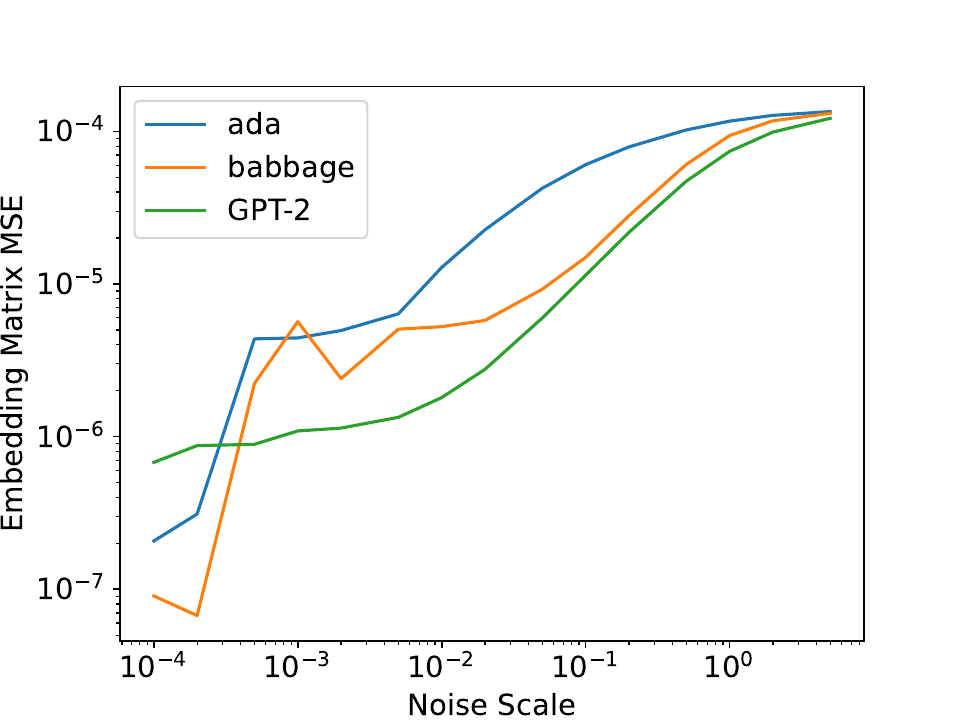}
    \captionof{subfigure}{RMSE of extracted embeddings at various noise variances.
    \label{subfig:noise-mse}}
  \end{minipage}
  \caption{In (a, b), recovering the embedding matrix dimension $h$ for Llama-7B at different levels of precision: 16-bit (default), 8-bit, and 4-bit. We observe no meaningful differences, with respect to our attack, at different levels of quantization. In (c), the RMSE between extracted embeddings as a function of the standard deviation of Gaussian noise added to the logits.}
  \label{fig:llama-7b-quantized-comparison}
\end{figure}

\begin{figure}
    \centering
    \includegraphics[width=\linewidth]{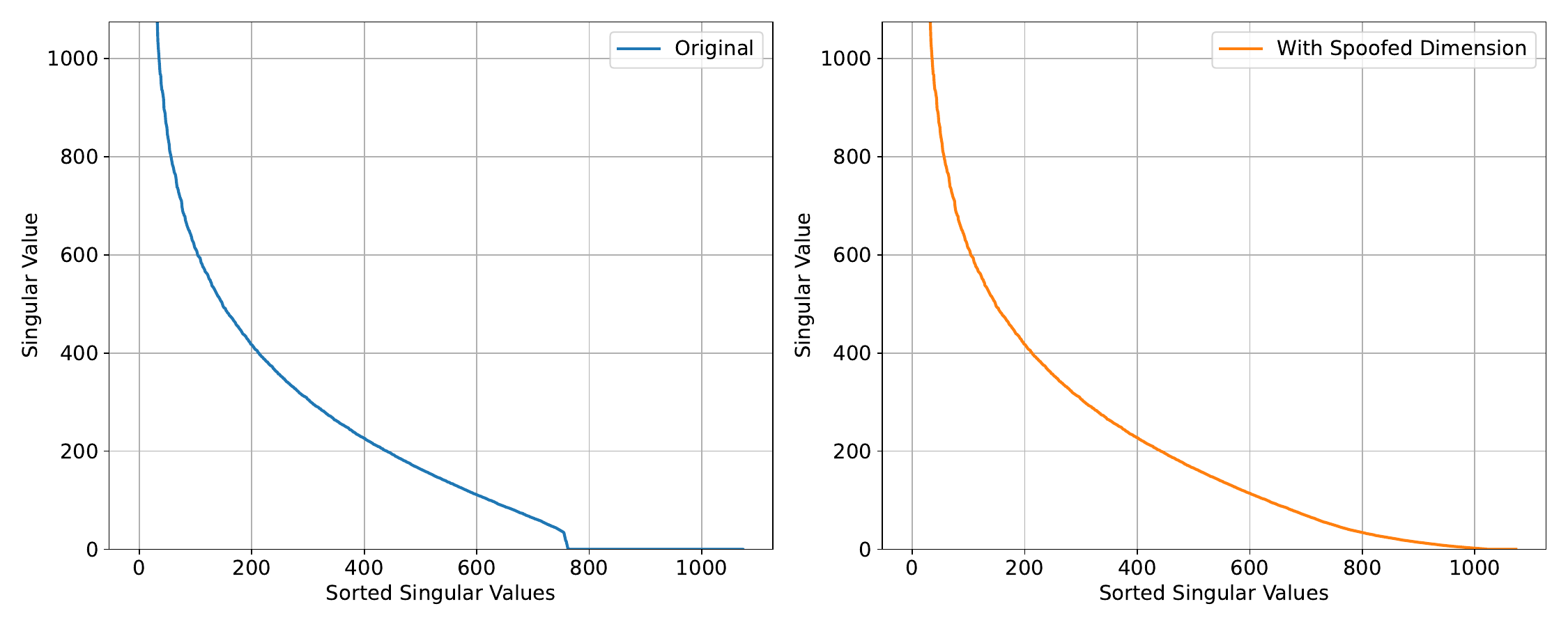}
  \vspace{-0.8cm}
  \caption{On the left, we plot the singular values that are extracted using our attack on GPT-2 small---the estimated hidden dimension is near 768. On the right, we post-hoc extend the dimensionality of the weight matrix to 1024, as described in Section~\ref{sec:defense}. This misleads the adversary into thinking the model is wider than it actually is.}
  \label{fig:gpt2_with_spoofing}
\end{figure}

\end{appendix}
\end{document}